\def\year{2022}\relax
\documentclass[letterpaper]{article} 
\usepackage{aaai22}  
\usepackage{times}  
\usepackage{helvet}  
\usepackage{courier}  
\usepackage[hyphens]{url}  
\usepackage{graphicx} 
\usepackage{natbib}  
\usepackage{caption} 
\DeclareCaptionStyle{ruled}{labelfont=normalfont,labelsep=colon,strut=off} 
\usepackage{cite}
\usepackage{amsmath,amssymb,amsfonts, amsthm}
\usepackage[vlined,ruled,linesnumbered]{algorithm2e}
\usepackage{textcomp}
\usepackage{xcolor}

\newcommand{\laurent}[1]{{\color{blue} [#1]}}

\newcommand{\nop}[1]{}
\newcommand{\todo}[1]{{\color{red} [#1]}}

\newtheorem{theorem}{Theorem}
\newtheorem{lem}{Lemma}

\newtheorem{proper}{Property}

\newtheorem{remark}{Remark}

\usepackage{subfigure}

\newcommand{\wmname}{CosWM}

\title{Cosine Model Watermarking Against Ensemble Distillation}
\author{
    Laurent Charette\textsuperscript{\rm 1}\equalcontrib,
    Lingyang Chu\textsuperscript{\rm 2}\equalcontrib, 
    Yizhou Chen\textsuperscript{\rm 3}, 
    Jian Pei\textsuperscript{\rm 3},
    Lanjun Wang\textsuperscript{\rm 4}, 
    Yong Zhang\textsuperscript{\rm 1}
    \\
}
\affiliations{
    \textsuperscript{\rm 1} Huawei Technologies Canada, 
    Burnaby, Canada $\mid$ \{laurent.charette, yong.zhang3\}@huawei.com\\
    
    \textsuperscript{\rm 2} McMaster University, 
    Hamilton, Canada $\mid$ chul9@mcmaster.ca\\
    
    \textsuperscript{\rm 3} Simon Fraser University, 
    Burnaby, Canada $\mid$ yca375@sfu.ca, jpei@cs.sfu.ca\\
    
    \textsuperscript{\rm 4} Tianjin University, 
    Tianjin, China $\mid$ wang.lanjun@outlook.com\\

}

\begin{document}

\maketitle

\begin{abstract}
Many model watermarking methods have been developed to prevent valuable deployed commercial models from being stealthily stolen by model distillations.
However, watermarks produced by most existing model watermarking methods can be easily evaded by ensemble distillation, because averaging the outputs of multiple ensembled models can significantly reduce or even erase the watermarks.
In this paper, we focus on tackling the challenging task of defending against ensemble distillation.
We propose a novel watermarking technique named \wmname{} to achieve outstanding model watermarking performance against ensemble distillation.
\wmname{} is not only elegant in design, but also comes with desirable theoretical guarantees.
Our extensive experiments on public data sets demonstrate the excellent performance of \wmname{} and its advantages over the state-of-the-art baselines.
\end{abstract}


\section{Introduction}


High-performance machine learning models are valuable assets of many large companies.
These models are typically deployed as web services where the outputs of models can be queried using public application programming interfaces (APIs)~\cite{ribeiro2015a}.

\nop{
that is, to offer public application programming interfaces (APIs) to deployed models while preventing the models from being stolen.
}

A major risk of deploying models through APIs is that the deployed models are easy to steal~\cite{tramer2016a}. By querying the outputs of a deployed model through its API, many model distillation methods \cite{orekondy2019a, jagielski2019a, papernot2017a} can be used to train a replicate model with comparable performance as the deployed model. 
Following the context of model distillation~\cite{hinton2015a}, a replicate model is called a \emph{student model}; and the deployed model is called a \emph{teacher model}.

\nop{
A replicate model built in this way is often said to be \emph{distilled} from the deployed model.
}

A model distillation process is often imperceptible because it queries APIs in the same way as a normal user~\cite{orekondy2019a}.
To protect teacher models from being stolen, one of the most effective ways is model watermarking~\cite{szyller2019a}. 
The key idea is to embed a unique watermark in a teacher model, such that a student model distilled from the teacher model will also carry the same watermark.
By checking the watermark, the owner of the teacher model can identify and reclaim ownership of a student model.

\nop{
Typical watermarks include digital signatures embedded in the parameters (i.e., model weights) of a deployed model~\todo{\cite{}} and adversarial images designed to trigger specified outputs of a deployed model~\todo{\cite{}}.
}

\nop{a straight-forward idea is to check whether a new model is distilled from a deployed model. This will help the owner of the deployed model to identify and reclaim the ownership of a replicate model.
}

\nop{ the outputs of a model queried through the APIs provide rich information to replicate the model.}

\nop{
With the advent of Machine Learning as a Service (MLaaS)~\cite{ribeiro2015a}, many companies are seeking to monetize high-performance machine learning models by providing public access to such models through an application programming interface (API).
}

\nop{
A client can query data to the API and obtain the output of the model in return.
For some popular services, such as image labeling or text translation, many companies offer very competent services in the market.
However, with similar services being put on the market, these models are vulnerable to be stolen by ensemble distillation attacks.
}



Some model watermarking methods have been proposed to identify student models produced by single model distillation~\cite{szyller2019a, lukas2019a, jia2021a}.
However, as we discuss in Section~\ref{sec:rw}, watermarks produced by these methods can be significantly weakened or even erased by \emph{ensemble distillation}~\cite{hinton2015a}, which uses the average of outputs queried from multiple different teacher models to train a replicate model.

Ensemble distillation has been well-demonstrated to be highly effective at compressing multiple large models into a small size student model with high performance~\cite{bucilua2006a4, ba2014a}.
On the other hand, the effectiveness of ensemble distillation also poses a critical threat to the safety of deployed models.

As shown by extensive experimental results in Section~\ref{sec:experiments}, ensemble distillation not only generates student models with better prediction performance, but also significantly reduces the effectiveness of existing model watermarking methods in identifying student models. 
As a result, accurately identifying student models produced by ensemble distillation is an emergent task with top priority to protect teacher models from being stolen.

In this paper, we focus on defending against model distillation, and we successfully tackle this task by introducing a novel model watermarking method named \wmname{}.
To the best of our knowledge, our method is the first model watermarking method with a theoretical guarantee to accurately identify student models produced by ensemble distillation.
We make the following contributions.

First, we present a novel method named \wmname{} that embeds a watermark as a cosine signal within the output of a teacher model.
Since the cosine signal is difficult to erase by averaging the outputs of multiple models, student models produced by ensemble distillation will still carry a strong watermark signal.

Second, under reasonable assumptions, we prove that a student model with a smaller training loss value will carry a stronger watermark signal.
This means a student model will have to carry a stronger watermark in order to achieve a better performance. Therefore, a student model intending to weaken the watermark will not be able to achieve a good performance.

Third, we also design \wmname{} to allow each teacher model to embed a unique watermark by projecting the cosine signal in different directions in the high-dimensional feature space of the teacher model.
In this way, owners of teacher models can independently identify their own watermarks from a student model.

Last, extensive experiment results demonstrate the outstanding performance of \wmname{} and its advantages over state-of-the-art methods.


\section{Related Works}
\label{sec:rw}

In this section, we introduce two major categories of model watermarking methods and discuss why these methods can be easily evaded by ensemble distillation.

The first category of methods~\cite{uchida2017a, rouhani2018a, adi2018a, zhang2018a, lemerrer2019a} aim to protect machine learning models from being exactly copied.
To produce a watermark, an effective idea is to embed a unique pattern by manipulating the values of parameters of the model to protect~\cite{uchida2017a, rouhani2018a}. If a protected model is exactly copied, the parameters of the copied model will carry the same pattern, which can be used as a watermark to identify the ownership of the copied model.
Another idea to produce a watermark is to use backdoor images that trigger prescribed model predictions~\cite{adi2018a, zhang2018a, lemerrer2019a}. The same backdoor image will trigger the same prescribed model prediction on an exactly copied model. Thus, the backdoor images are also effective in identifying exactly copied models.

The above methods focus on identifying exactly copied models, but they cannot be straight-forwardly extended to identify a student model produced by ensemble distillation~\cite{hinton2015a}. Because the model parameters of the student model can be substantially different from the teacher model; and simple backdoor images of the teacher model are often not transferable to the student model, that is, the backdoor images may not trigger the prescribed model prediction on the student model~\cite{lukas2019a}.

The second category of methods aim to identify student models that are distilled from a single teacher model by single model distillation~\cite{tramer2016a}.

\nop{
One of the first defense mechanism designs against distillation to emerge was PRADA~\cite{juuti2019a}.
This method would identify a distillation attempt by monitoring the distribution of query inputs for each client.
If queries significantly differ from the natural data distribution a client would flagged and denied access to the service.
}

\nop{
This method, while having the potential of preventing distillation altogether by an adversary, can be evaded by using a querying strategy more in line with the natural data distribution~\cite{orekondy2019a}.
}

PRADA~\cite{juuti2019a} is designed to identify model distillations using synthetic queries that tend to be out-of-distribution. It analyzes the distribution of API queries and detects potential distillation activities when the distribution of queries deviates from the benign distribution.
However, it is not effective in identifying the queries launched by ensemble distillations, because these queries are mostly natural queries that are not out-of-distribution.

\nop{
\laurent{However, it can be easily evaded if adversaries adapt their querying strategies by using more natural data~\cite{orekondy2019a}.}
}
\nop{
However, it is not effective in identifying the queries launched by ensemble distillations, because these queries are mostly natural queries that are not out-of-distribution.
}

Another typical idea is to produce transferable backdoor images that are likely to trigger the same prescribed model prediction on both the teacher model and the student model.
DAWN~\cite{szyller2019a} generates transferable backdoor images by dynamically changing the outputs of the API of a protected teacher model on a small subset of querying images.
Fingerprinting~\cite{lukas2019a} makes backdoor images more transferable by finding common adversarial images that trigger the same adversarial prediction on a teacher model and any student model distilled from the teacher model.
Entangled Watermarks~\cite{jia2021a} forces a teacher model to learn features for classifying data sampled from the legitimate data and watermarked data.

The above methods are effective in identifying student models produced by single model distillation, but they cannot accurately identify student models produced by ensemble distillation. 

The reason is that, when an ensemble distillation averages the outputs of a watermarked teacher model and multiple other teacher models without a watermark, the prescribed model predictions of the watermarked teacher model will be weakened or even erased by the normal predictions of the other teacher models.
If multiple watermarked teacher models are used for ensemble distillation, the prescribed model prediction of one teacher model can still be weakened or erased when averaged with the predictions of the other teacher models, because the prescribed model predictions of different teacher models are not consistent with each other.

\nop{
a protected teacher model is  the prescribed model predictions of a backdoor image can be overwritten by 

The reason is that  However, 

The reason is that all three methods depend on watermark output predictions being different from that of other models.
However, in an ensemble the output prediction of the watermarked model will be outnumbered with the other models' outputs, which will erase most traces of the watermark in the ensemble outputs to be learned by the distilled model.
}

\nop{
\wmname{} distinguishes itself from the existing watermarking methods by being specifically designed to be robust against ensemble distillation.
\nop{The key difference uses a slight output modification in the shape of a cosine signal, and the trace of this signal is preserved after aggregation, even if the watermarked model is outnumbered.}
The cosine signal injected in the watermarked model outputs remains in the output of the ensemble model after aggregation.
The distilled model, aiming to match the output of the ensemble model, will also bear enough traces of the cosine signal to be reliably extracted by the extraction process.
}

The proposed \wmname{} method is substantially different from the other watermarking methods~\cite{szyller2019a, lukas2019a, jia2021a}.
The watermark of \wmname{} is produced by coupling a cosine signal with the output function of a protected teacher model.
As proved in Theorem~\ref{bound} and demonstrated by extensive experiments in Section~\ref{sec:experiments}, when an ensemble distillation averages the outputs of multiple teacher models, the embedded cosine signal will persist.
As a result, the watermarks produced by \wmname{} are highly effective in identifying student models produced by ensemble distillation.

\nop{
This allows the watermark produced by \wmname{} to be transferred to the distilled model, even if the watermarked model is outnumbered in the ensemble.
}

\section{Problem Definition}
\label{sec:pd}

\nop{
In this section, we first introduce a procedure of model replication using an ensemble of teacher models and then describe a watermarking task, in which 
a traceable mark is added into a teacher model and only reduces the model's accuracy within a given range.

\subsection{Model Replication via Ensemble Distillation}
}
Ensemble methods, such as bagging~\cite{buhlmann2002a}, aggregate the probability predictions of all models in an ensemble to create a more accurate model on average.
Ensemble models and distillation have been applied jointly since the first seminal studies on distillation~\cite{bucilua2006a4, ba2014a, hinton2015a}.
These distillation methods use a combination of KL loss \cite{kullback1951a} and cross-entropy loss \cite{bishop2006} in the training process. Cross-entropy loss requires ground truth labels.
Some recent state-of-the-art distillation methods~\cite{vongkulbhisal2019a, shen2020a} only use KL loss, and thus can work without access to the ground truth values.
This allows adversaries to replicate high performance models using ensemble model distillation and without ground truth labels.

Technically, let $\mathcal R=\{R_1, \ldots, R_{N}\}$ be a set of $N$ models trained to perform the same $m$-class classification task.
Each model $R_i$ outputs a probability prediction vector $R_i(\mathbf{x})$ on an input sample $\mathbf{x}\in\mathbb{R}^n$.
An adversary may effectively build an ensemble model by querying an unlabeled data set $X^S=\{\mathbf{x}^1,\ldots,\mathbf{x}^L\} $ to each model $R_1, \ldots, R_{N}$ 
and averaging the outputs, i.e., $\mathbf{\bar q}^l = \frac{1}{N}\sum_{i=1}^N R_i(\mathbf{x}^l)$ for $l=1,\ldots,L$.
The averaged output $\mathbf{\bar q}^l$ can then be used as the soft pseudo labels to train a student model $S$. 


We now formulate the task of watermarking against distillation from ensembles. Assume a model $R$ to be protected and the watermarked version $w(R)$, where $w(\cdot)$ is a watermarking function. Denote by $h(R)$ a function measuring the accuracy of model $R$ (on a given test data set) and by $g(R)$ a function measuring the strength of the watermark signal in model $R$. 

Let $S$ be an arbitrary model that is replicated from an ensemble distillation using $w(R)$ as a teacher.  $S$ may use some additional other teacher models.  Let $S'$ be another arbitrary model that is replicated from an ensemble distillation where $w(R)$ is not a teacher. The \textbf{task of model watermarking} is to design watermarking function $w(\cdot)$ such that it meets two requirements.  First, the accuracy loss in watermarking is within a specified tolerance range $\alpha>0$, i.e., $h(R) - h(w(R)) \leq \alpha$.  Second, the watermark signal model in $S$ is stronger than that in $S'$, i.e., $g(S) > g(S')$.

\nop{

Suppose there are $M$ suspected models, $S_1, \ldots, S_M$
possibly replicated from the ensemble distillation using $M$ different model sets $\mathcal R_1, \ldots, \mathcal R_M$. Without loss of generality, we assume the first $L$ model sets all contain the watermarked model $R_w$, i.e., $R_w \in \mathcal R_i$ for $i=1,\ldots,L$, while each of the remaining $M-L$ model sets do not contain the watermarked model $R_w$, i.e., $R_w \notin \mathcal R_i$ for $i=L+1,\ldots,M$.
Then, the task of model watermarking requires that, after applying a watermark function $w(\cdot)$ on model $R_0$, the accuracy of the watermarked model $R_w$ is dropped by less than a given tolerance $\alpha$, i.e.,
\[
|h(R_0) - h(R_w)| \leq \alpha
\]
and the first $r$ suspected models have stronger watermark signal than the rest $m-r$ suspected models, i.e.,
\[
\min_{1\leq i \leq L} g(S_i) \geq \max_{L+1\leq i \leq M} g(S_i).
\]
}

\section{\wmname{}}
\label{sec:polytope}

In this section, we present our watermarking method \wmname{}. We first explain the intuition of our method.  Then, we develop our watermarking framework to embed a periodic signal to a teacher model. Third, we describe how the embedded signal can be extracted from a student model learned using a watermarked teacher model.  Next, we provide strong theoretical results to justify our design. Last, we discuss possible extensions to ensembles containing multiple watermarked models.

\subsection{Intuitions}

The main idea of \wmname{} is to introduce a perturbation to the output of a teacher model. This perturbation is transferred onto a student model distilled from the teacher model and remains detectable with access to the output of the student model. 

\begin{figure}[t]
\begin{center}
\includegraphics[width=0.47\textwidth]{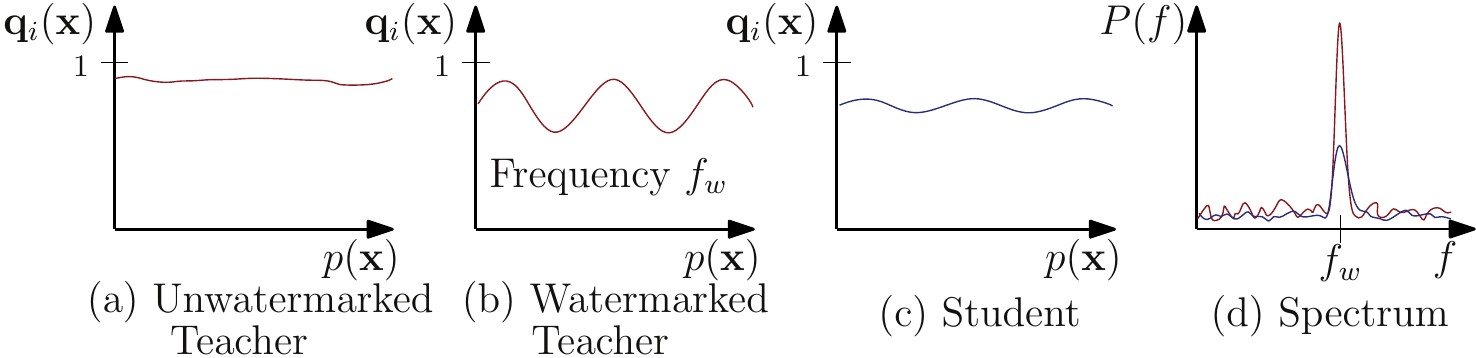}
\end{center}
\caption{The idea of \wmname{}, where $\mathbf{q}_i(\mathbf{x})$ is an model output component for image $\mathbf{x}$, $p(\mathbf{x})$ is a projection as described in equation (\ref{eq:projection}), $f$ and $P(f)$ are frequency and power spectrum values of a $p(\mathbf{x})$-$ \mathbf{q}_i(\mathbf{x})$ graph.}
\label{fig:keyidea}
\end{figure}

The idea is illustrated in Figure~\ref{fig:keyidea}.  Let $R$ be a model to be watermarked and $\mathbf{q}=R(\mathbf{x})$ be the output of the model $R$ on input $\mathbf{x}$.
We also convert $\mathbf{x}$ into a number $p(\mathbf{x})$ in a finite range. We can select a class $i^*$ and use the model prediction output $\mathbf{q}_{i^*}$ on the class to load our watermark. Let $\mathbf{q}_{i^*}(\mathbf{x})$ be the $i^*$-th element of vector $R(\mathbf{x})$.
Figure~\ref{fig:keyidea}(a) plots $(\mathbf{q}_{i^*}(\mathbf{x}), p(\mathbf{x}))$ without any added watermark signal. After adding a periodic perturbation $\phi(p(\mathbf{x}))$ of frequency $f_w$ to the output of $R$, the new output $\mathbf{q}_{i^*}(\mathbf{x})$ demonstrates some oscillations, as shown in Figure~\ref{fig:keyidea}(b). 
We keep the perturbation small enough so that the model predictions are mostly unaffected and the effect of the watermark on the model's performance is minimal.

A student model trying to replicate the behavior of the teacher model passively features a similar oscillation at the same frequency $f_w$.
In addition, even with the averaging effect of an ensemble of teacher models on the outputs, the periodic signal should still be present in some form. 
Since the averaging is linear, the amplitude is diminished by a factor of the number of the ensemble models as shown in Figure~\ref{fig:keyidea}(c). 
By applying a Fourier transform, the perturbation can be re-identified by the presence of a peak in the power spectrum at the frequency $f_w$ as shown in  Figure~\ref{fig:keyidea}(d).

\nop{
Next we describe how we apply the \wmname{} watermarking framework to embed a periodic signal to a teacher model in Subsection~\ref{embedding} and extract the signal from a suspected student model in Subsection~\ref{extraction}.
We conclude this section with some theoretical results in Subsection~\ref{sec:theoretical}, and a short discussion on the extension to ensembles containing multiple watermarked models in Subsection~\ref{extension_multi}.
}
\subsection{Embedding Watermarks to a Teacher Model}
\label{embedding}


Normally, an output $\mathbf{q}$ of a model $R$ on a given data point $\mathbf{x}$ is calculated from the softmax of the logits $\mathbf{z}\in\mathbb{R}^m$, i.e.,
\begin{equation}\label{eq:softmax}
    \mathbf{q}_i = \frac{e^{\mathbf{z}_i}}{\sum_{j=1}^m e^{\mathbf{z}_j}}, \ \mbox{for} \ i=1, ..., m,
\end{equation}
where $\mathbf{z}$ is a function of $\mathbf{x}$, and $\mathbf{q}_i$ is the $i$-th element of vector $\mathbf{q}$.  As a result, the output $\mathbf{q}$ has the following property.

\begin{proper}\label{prop:softmax}
Let $\mathbf{q}$ be a softmax of the logit output $\mathbf{z}$ of a model $R$. Then,
\begin{enumerate}
\label{softmax-properties}
    \item $0 \leq \mathbf{q}_i \leq 1$ for $i=1, \ldots, m$,
    \item $\sum_{i=1}^m \mathbf{q}_i = 1$.
\end{enumerate}
\end{proper}

We want to substitute $\mathbf{q}$ in the model inference by a modified output $\mathbf{\hat{q}}\in\mathbb{R}^m$ which features the periodic signal and satisfies Property~\ref{prop:softmax}.
However, only modifying $\mathbf{q}$ in the model inference by itself may degrade the performance of the model, and the loss in accuracy cannot be bounded. In order to mitigate this effect, we also use the modified output $\mathbf{\hat{q}}$ in training $R$. That is, we use $\mathbf{\hat{q}}$ to compute cross entropy loss in the training process.

To embed watermarks, we first define a watermark key $K$ that consists of a target class $i^*\in \{1,\ldots,m\}$, an angular frequency $f_w\in\mathbb{R}$, and a random unit projection vector $\mathbf{v}\in\mathbb{R}^n$, i.e., $K=(i^*, f_w,  \mathbf{v})$. Using $K$, we define a periodic signal function
\begin{equation}\label{eq:signal_function}
\mathbf{a}_i( \mathbf{x}) = \left\{
\begin{aligned}
    &\cos \left( f_w p(\mathbf{x}) \right), & \mbox{if} \ i = i^*, \\
    &\cos \left( f_w p(\mathbf{x}) + \pi \right), & \mbox{otherwise},
\end{aligned}
\right.
\end{equation} 
for $i \in \{1, \ldots, m\}$, where 
\begin{equation}\label{eq:projection}
   p(\mathbf{x})  = \mathbf{v}^{\intercal} \mathbf{x}.
\end{equation}

We consider single-frequency signals in this work and we plan to study watermark signals with mixed frequencies in our future work.
We adopt linear projections since they are simple one-dimensional functions of input data and can easily form a high-dimensional function space.
This leads to a large-dimensional space to select $\mathbf{v}$ from, and generally little interference between two arbitrary choices of $\mathbf{v}$.
As a consequence, we get a large choice of possible watermarks, and each watermark is concealed to adversaries trying to source back the signal with arbitrary projections.

We inject the periodic signal into output $\mathbf{q}$ to obtain $\mathbf{\hat{q}}$ as follows. For $i \in \{1, \ldots, m\}$, 
\begin{equation}\label{eq:modified_softmax}
   \mathbf{\hat{q}}_i =  \left\{
\begin{aligned}
& \frac{\mathbf{q}_i + \varepsilon (1 + \mathbf{a}_i( \mathbf{x}) )}{1 + 2 \varepsilon}, &   \mbox{if} \ i = i^*, \\
& \frac{\mathbf{q}_i + \frac{\varepsilon (1 + \mathbf{a}_i( \mathbf{x}) )}{m-1}}{1 + 2 \varepsilon}, &  \mbox{otherwise},
\end{aligned}
\right.
\end{equation}
where  $\varepsilon$ is an amplitude component for the watermark periodic signal.
As proved in Lemma~\ref{lemma-1} in Appendix~\ref{proof-lemma}, the modified output $\mathbf{\hat{q}}$ still satisfies both requirements in Property~\ref{softmax-properties}.
Therefore, it is natural to replace $\mathbf{q}$ by $\mathbf{\hat{q}}$ in inference. 

Nevertheless, if we only modify $\mathbf{q}$ into $\mathbf{\hat{q}}$ in inference, the inference performance can be degraded by this perturbation.
Since the modified output satisfies Property~\ref{softmax-properties}, we can use it in training as well to compensate for the potential performance drop.
To do that, we directly replace $\mathbf{q}$ by $\mathbf{\hat{q}}$ in the cross-entropy loss function.
Specifically, for a data point $\mathbf{x}$ with one-hot encoding true label $\mathbf{y}^t\in\mathbb{R}^m$, the cross-entropy loss during training can be replaced by
\begin{equation}
\label{eq:modified_loss}
L_{CE, wm} = -\sum_{j=1}^m \mathbf{y}^t_j \log \left(\mathbf{\hat{q}}_j\right).
\end{equation}

The model $R_w$ trained as such carries the watermark. By directly modifying the output, we ensure that the signal is present in every output, even for input data not used during training.
This generally results in a clear signal function in the output of the teacher model $R_w$ that is harder to conceal by noise caused by distillation training or by dampening due to ensemble averaging. 

 
 

\subsection{Extracting Signals in Student Models}
\label{extraction}

\nop{
In this subsection, we present our algorithm for ranking a list of suspected student models which may be distilled from an ensemble featuring the watermarked teacher model $R_w$.

When a collection of suspected models is presented, we test the watermark of model $R_w$ on each model and rank their extracted signal strength for the watermark's frequency $f_w$.
Then the models with the higher signal strength are more likely to have been distilled from an ensemble featuring $R_w$.
}

Let $S$ be a student model that is suspected of being distilled from a watermarked model $R_w$ or multiple ensembled teacher models including $R_w$. 
To extract the possible watermark from $S$, we need to query $S$ with a sample of student training data $\widetilde X^S=\{\mathbf{x}^1,\ldots,\mathbf{x}^{\widetilde L}\}$. 
According to~\cite{szyller2019a}, the owner of a teacher model can easily obtain $\widetilde X^S$ because the owner may store any query input sent by an adversary to the API.


Let the output of model $S$ on the input data $\widetilde X^S$ be $\widetilde Q^S =\{\mathbf{q}^1,\ldots,\mathbf{q}^{\widetilde L}\}$, where $\mathbf{q}^{l} \in\mathbb{R}^m$ for $l=1,\ldots, {\widetilde L}$.
For every pair $(\mathbf{x}^l, \mathbf{q}^l)$, we extract a pair of results $(\mathbf{p}_l, \mathbf{q}_{i^*}^l)$, where $\mathbf{p}_l = \mathbf{v}^{\intercal} \mathbf{x}^l$ as per Equation~\eqref{eq:projection}, $\mathbf{v}$ is in the watermark key of $R_w$ and $i^*$ is the target class when embedding watermarks to $R_w$.
We filter out the pairs $(\mathbf{p}_l, \mathbf{q}_{i^*}^l)$ with $\mathbf{q}_{i^*}^l \leq q_{min}$ in order to remove outputs with low confidence, where the threshold value $q_{min}$ is a constant parameter of the extraction process.
The surviving pairs are re-indexed into a set $\widetilde{D}^S = \{(\mathbf{p}_l, \mathbf{q}_{i^*}^l)\}_{l = 1, \ldots, \widetilde{M}}$, where $\widetilde{M}$ is the number of remaining pairs.
These surviving pairs $(\mathbf{p}_l, \mathbf{q}_{i^*}^l) \in \widetilde{D}^S$ are then used to compute the Fourier power spectrum, for evenly spaced frequency values spanning a large interval containing the frequency $f$.

To approximate the power spectrum, we use the Lomb-Scargle periodogram method~\cite{scargle1982studies}, which allows one to approximate the power spectrum $P(f)$ at frequency $f$ using unevenly sampled data. 
We give the formal definition of $P(f)$ in Section~\ref{sec:theoretical} when we analyze the theoretical bounds of $P(f)$.
Due to noise in the model outputs, it is preferable to have more sample pairs in $\widetilde{D}^S$ than the few required to detect a pure cosine signal.
In our experience, we reliably detect a watermark signal using 100 pairs for a single watermarked model and 1,000 pairs for an 8-model ensemble.

To measure the signal strength of the watermark, we define a maximum frequency $F$ and a window $\left[ f_w-\frac{\delta}{2}, f_w+\frac{\delta}{2} \right]$, where $\delta$ is a parameter for the width of the window and $f_w$ is the frequency in watermark key of $R_w$. 
Then, we calculate $P_{signal}$ and $P_{noise}$ by averaging spectrum values $P(f)$ on frequencies inside and outside the window, i.e.,
$P_{signal}=\frac{1}{\delta}\int_{f_w-\frac{\delta}{2}}^{f_w+\frac{\delta}{2}} P(f) df$ and $P_{noise} = \frac{1}{F - \delta} \left[\int_{0}^{f_w-\frac{\delta}{2}} P(f) df + \int_{f_w+\frac{\delta}{2}}^{F} P(f) df \right]$, respectively.
We use the signal-to-noise ratio to measure the signal strength of the watermark,
i.e.,
\begin{equation}
\label{psnr}
    P_{snr} = P_{signal} / P_{noise}.
\end{equation}

The algorithm is summarized in Algorithm~\ref{alg:extracting}.

\begin{algorithm}[t]
\SetAlgoLined
\SetKwInOut{Input}{Inputs}
\SetKwInOut{Output}{Output}
\Input{A suspected model $S$, \\
Samples $\widetilde X_S$ of the training data of $S$, \\
A watermark key $K=(i^*, f_w, \mathbf{v})$ of the \\ watermarked model $R_w$, \\
Filtering threshold value $q_{min}$.}
\Output{Signal strength.}
 Query $\widetilde X_S$ to $S$ and obtain outputs $\widetilde Q^S =\{\mathbf{q}^1,\ldots,\mathbf{q}^{\widetilde L}\}$.
 
 Compute projections $\mathbf{p}_l = \mathbf{v}^{\intercal} \cdot \mathbf{x}^l$, for $l=1,\ldots,\widetilde L$.
 
 Filter out outputs where $\mathbf{q}_{i^*}^l \leq q_{min}$, remaining pairs form the set $\widetilde{D}^S = \{(\mathbf{p}_l, \mathbf{q}_{i^*}^l)\}_{l = 1, \ldots, \widetilde{M}}$.
 
 Compute the Lomb-Scargle periodogram from the pairs $(\mathbf{p}_l, \mathbf{q}_{i^*}^l)$ in $\widetilde{D}^S$.
 
 Compute $P_{signal}$ and $P_{noise}$ by averaging spectrum values on frequencies inside and outside the window $\left[ f_w-\frac{\delta}{2}, f_w+\frac{\delta}{2} \right]$, respectively.
 
 Compute $P_{snr} = P_{signal} / P_{noise}$.
 
 \Return Signal strength $P_{snr}$.
 
 \caption{Extracting signal in a model}
 \label{alg:extracting}
\end{algorithm}


\subsection{Theoretical Analysis}
\label{sec:theoretical}

Here, we analyze the signal strength of $P_{signal}$ and  $P_{noise}$ and provide theoretical bounds for the power spectrum $P(f)$. Let us first recall two results from~\cite{scargle1982studies}.

Given a paired data set $D=\{(\mathbf{a}^l, \mathbf{b}_l) \in \mathbb{R}^n \times \mathbb{R}, l=1,\ldots,L\}$, an angular frequency $f$, a projection vector $\mathbf{v}$, and a sinusoidal function 
$
        s(\mathbf{x}) = \alpha + \beta\cos(f\mathbf{v}^{\intercal} \mathbf{x} + \gamma),
$ 
    where $\alpha$, $\beta$ and $\gamma$ are the parameters of $s(\mathbf{x})$, the \emph{best fitting points} $\mathbf{s}^*(D)$ for this paired data are
     \begin{equation}
     \label{best_fit}
      [ \mathbf{s}^*(D)]_l = \alpha^* + \beta^*\cos(f\mathbf{v}^{\intercal} \mathbf{a}^l + \gamma^*) \ \mbox{for}\ l=1,\ldots, L,
    \end{equation}
    where the parameters $\alpha^*$, $\beta^*$, $\gamma^*$ minimize the square error 
$
 \chi^2_f(D)= \sum_{l=1}^L [\mathbf{b}_l - s(\mathbf{a}^l)]^2.
$

Moreover, given a paired data set $D=\{(\mathbf{a}^l, \mathbf{b}_l) \in \mathbb{R}^n \times \mathbb{R}, l=1,\ldots,L\}$ and a frequency $f$, the \emph{unnormalized Lomb-Scargle periodogram} can be written as
\begin{equation}
\label{eq:LSdef}
P_D(f) = \frac{1}{2}\left[\chi_0^2(D) - \chi_f^2(D)\right],
\end{equation}
where $\chi_0^2(D) $ is the square error of the best constant fit to $\mathbf{b}_1,\ldots,\mathbf{b}_L$.

Now we are ready to give a theoretical bound on $P_D(f)$ for the output of the student model.


\begin{theorem}
\label{bound}
Suppose there are $N$ teacher models $R_1,\ldots, R_N$. Without loss of generality, let $R_1$ be a watermarked teacher model with watermark key $K=(i^*, f_w, \mathbf{v})$, and $S$ a student model distilled from an ensemble model of $R_1,\ldots, R_N$ on the student training data $X^S$.
Let $\widetilde{X}^S=\{\mathbf{x}^1,\ldots,\mathbf{x}^L\}$ be a sample subset of $X^S$.
Let $\mathbf{\hat q}^l = R_1(\mathbf{x}^l)$ be the output of model $R_1$, $\mathbf{\widetilde q}^l = \frac{1}{N-1}\sum_{i=2}^N R_i(\mathbf{x}^l)$ be the output of the ensemble model of $R_2, \ldots, R_N$, $\mathbf{\bar q}^l = \frac{1}{N}(\mathbf{\hat q}^l  + (N-1)\mathbf{\widetilde q}^l)$ be the output of the ensemble model of $R_1,\ldots, R_N$, and $\mathbf{q}^l = S(\mathbf{x}^l)$ the output of $S$ for the training data point $\mathbf{x}^l$.
Let $\hat{D} = \{(\mathbf{x}^l, \mathbf{\hat q}^l_{i^*}), l= 1,\ldots, L\}$, $\widetilde{D} = \{(\mathbf{x}^l, \mathbf{\widetilde q}^l_{i^*}), l= 1,\ldots, L\}$, $\bar{D} = \{(\mathbf{x}^l, \mathbf{\bar q}^l_{i^*}), l= 1,\ldots, L\}$ and  $D = \{(\mathbf{x}^l, \mathbf{q}^l_{i^*}), l= 1,\ldots, L\}$ be paired data sets.
Then, the unnormalized Lomb-Scargle periodogram value $P_D(f)$ for the student output at angular frequency $f$ has the following bounds
\begin{equation}\label{eq:bound}
 \frac{1}{2}{\left[ \chi_0^2(D){-}\tau_1{+}L_{se} \right]} {\geq} P_D(f) {\geq} \frac{1}{2}{\left[ \chi_0^2 (D){-}\tau_2{-}L_{se} \right]},
\end{equation}
where
\begin{align*}
\tau_1 = {\chi}^2_f(\bar{D}), \ 
&\tau_2=  \frac{1}{N^2}{\chi}^2_f(\hat{D}) +
\left(\frac{N-1}{N}\right)^2 {\chi}^2_f(\widetilde{D}),\\
&L_{se} = \sum_{l=1}^L \left(\mathbf{\bar q}^l_{i^*} -  \mathbf{q}^l_{i^*} \right)^2.\\
\end{align*}
\proof See Appendix~\ref{proof}.
\end{theorem}

Theorem~\ref{bound} provides several insights.
\begin{remark}
When a student model is well trained by a teacher model, $L_{se}$ is generally small. 
\end{remark}

\begin{remark} 
\label{rem:remark2}
Consider the case where $f=f_w$.
If we choose our sample $\widetilde{X}^S$ with high confidence output scores on the $i^*$-th class, for example by filtering as described in Algorithm~\ref{alg:extracting}, 
${\chi}^2_{f_w}(\hat{D})$ should be small enough to be negligible by our watermark design in the teacher model.
We then discuss the following two cases.

\textbf{Case I}: When $N=1$, there is only one watermarked teacher to distill a student model. Then, after neglecting ${\chi}^2_{f_w}(\hat{D})$, the left inequality of Equation~\eqref{eq:bound} becomes 
\begin{equation*}
  P_D(f_w) \geq  \frac{1}{2} \left[ \chi_0^2(D)  -L_{se}\right] .
\end{equation*}
This implies that we can observe a significant signal for the output of the 
student model at frequency $f_w$ when the output of the student model is close to that of the teacher model.

\textbf{Case II}: When $N\neq 1$, since there is no sinusoidal signal in $\widetilde{\mathbf{q}}^l_{i^*}$, for $l=1,\ldots,L$, and the sinusoidal signal in $\bar{\mathbf{q}}^l_{i^*}$, for $l=1,\ldots,L$ is, proportional to $\frac{\varepsilon}{N}$, $\tau_2$ increases as $N$ increases. However, to keep the watermark signal significant in the output of the student model, 
one can increase the watermark signal amplitude $\varepsilon$ in the teacher model $R_1$, which indirectly increases $\chi_0^2(D)$. This is due to the fact that if $\varepsilon$ increases, $\chi_0^2(\hat{D})$ also increases.
Since $L_{se}$ is small when a student model is well trained by the teacher model, $\chi_0^2(D)$ increases as well.
This implies that we can detect the watermark in the output of the student model at frequency $f_w$ by increasing the watermark signal in the teacher model $R_1$ when $N$ is large. We validate this observation in our experiments in Section~\ref{sec:single}.
\end{remark}

\begin{remark}
\label{rem:remark3}
When $f\neq f_w$, since there is no sinusoidal signal at frequency $f$ in $\hat{\mathbf{q}}^l_{i^*}$, $\widetilde{\mathbf{q}}^l_{i^*}$, and $\bar{\mathbf{q}}^l_{i^*}$ for $l=1,\ldots,L$, $\chi^2_f(\hat{D})$, $\chi^2_f(\widetilde{D})$ and  $\chi^2_f(\bar{D})$ are generally large. Thus, the values of both sides of the inequality in Equation~\eqref{eq:bound} are small, which implies that there is no sinusoidal signal for the output of the student model at frequency $f\neq f_w$.
\end{remark}

\subsection{Multiple Watermarked Teacher Models}
\label{extension_multi}
Consider a student model trained on the output of an ensemble model that consists of two or more teacher models with watermarks.  Can those watermarks be detected in the student model?

We argue that it should be possible to extract each signal if the watermark keys are different.
The reason for this is that a signal embedded using watermark key $K_1 = (i_1, f_1, \mathbf{v}^1)$ appears as noise for an independent watermark $K_2 = (i_2, f_2, \mathbf{v}^2)$.
Since noise has low overall spectrum values, the resulting ensemble output spectrum will be similar to an ensemble with only one watermarked model.
Therefore, each signal should be detectable using its respective key.
This highlights the importance that $\mathbf{v}$ should preferably be a high dimensional vector that can provide more independent random choices for the watermark key $K$.


\begin{figure*}[t]%
    \centering
    \subfigure[Unwatermarked]{
    \includegraphics[width = 0.32\textwidth]{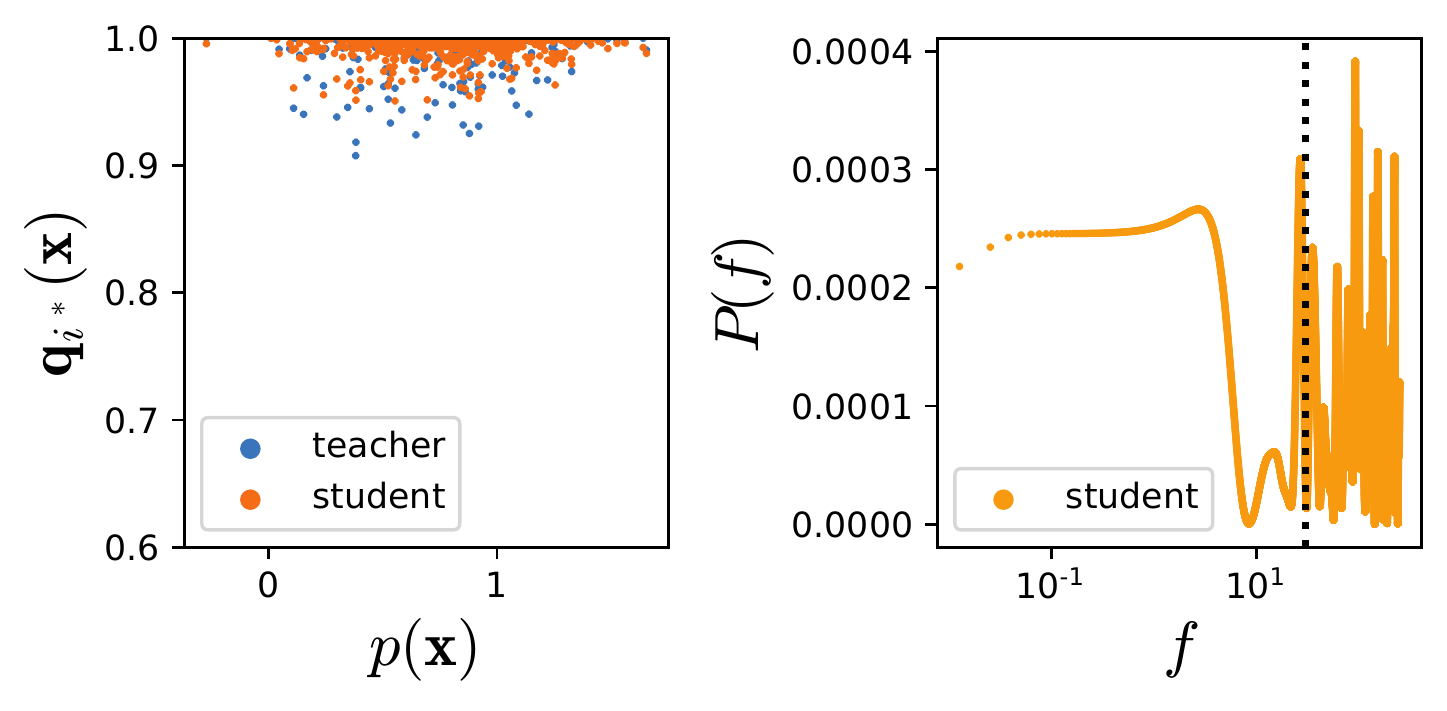}
    }
    \subfigure[Watermarked -- matching projection]{
    \includegraphics[width = 0.32\textwidth]{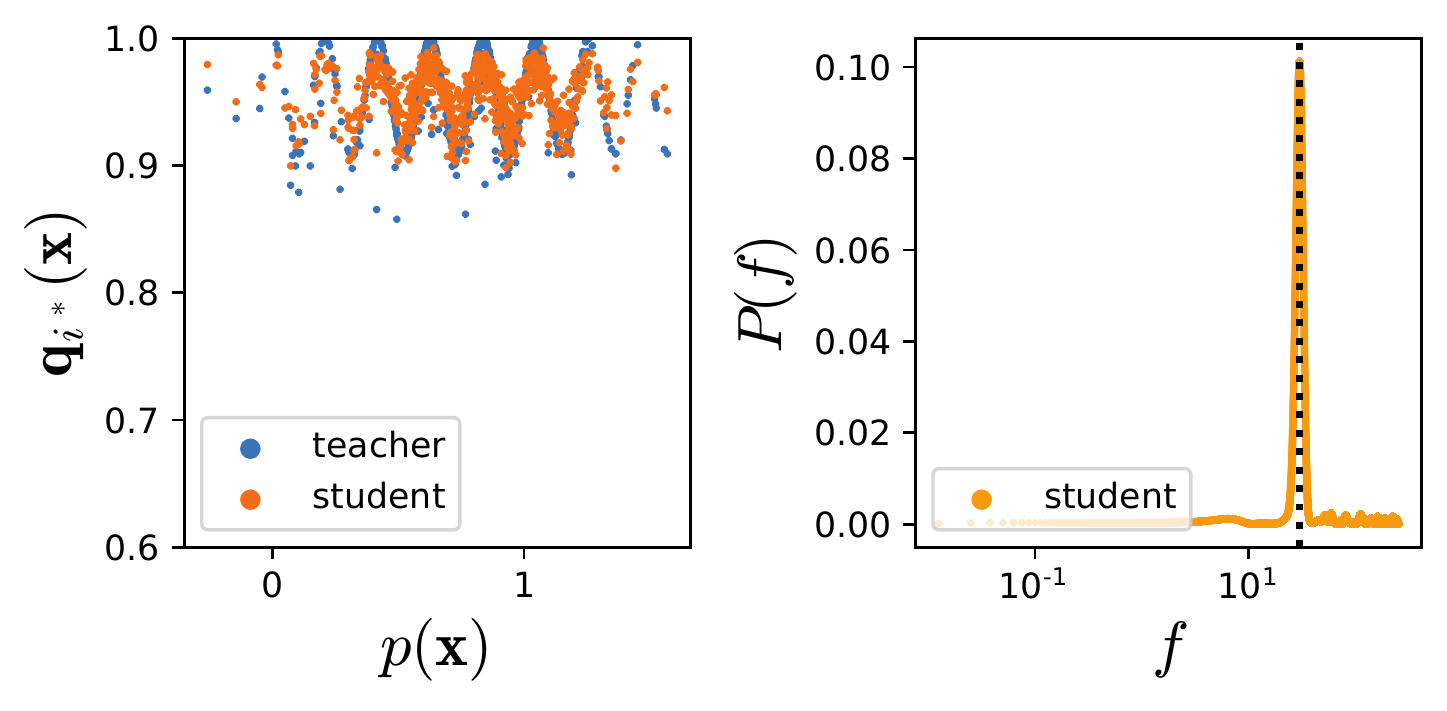}
    }
    \subfigure[Watermarked -- non-matching projection]{
    \includegraphics[width = 0.32\textwidth]{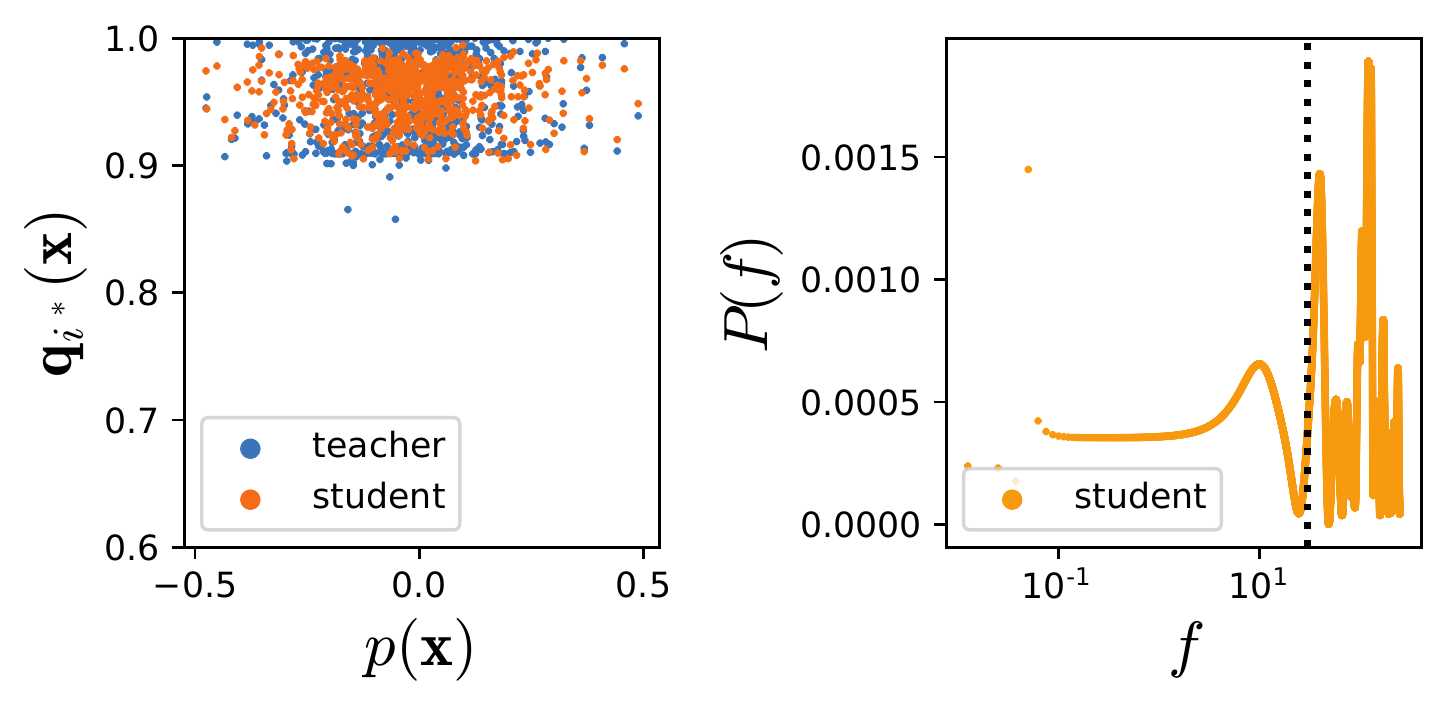}
    }
    \caption{
    A case study of the watermarking mechanism in \wmname{}.
    The black vertical line indicates $f = 30.0$. In each subgraph, the left plots the target class output $\mathbf{q}_i(\mathbf{x})$ of the teacher model and the student model as a function of projection value $p(\mathbf{x})$, and the right plots the power spectrum value $P(f)$ for the output of the student model as a function of frequency $f$.
    }
    \label{fig:case_study}
\end{figure*}

\section{Experiments}
\label{sec:experiments}

In this section, we evaluate the performance of \wmname{} on the model watermarking task. We first describe the settings and data sets in Section~\ref{sec:settings}.  Then we present a case study to demonstrate the working process of \wmname{} in Section~\ref{sec:case}. We compare the performance of all the methods in two scenarios in Sections~\ref{sec:single} and~\ref{sec:multiple}. We analyze the effect of the amplitude parameter  $\varepsilon$ and the signal frequency parameter $f_w$ on the performance of \wmname{} in Appendices~\ref{appx:ens_size_graphs} and~\ref{appx:frequency}, respectively. We also analyze the effects of using ground truth labels during distillation in Appendix~\ref{appx:truelabel}.

\subsection{Experiment Settings and Data Sets}
\label{sec:settings}

We compare \wmname{} with two state-of-the-art methods, DAWN \cite{szyller2019a} and Fingerprinting \cite{lukas2019a}.
We implement \wmname{} and replicate DAWN in PyTorch 1.3.
The Fingerprinting code is provided by the authors of the corresponding paper \cite{lukas2019a} and is implemented in Keras using a TensorFlow v2.1 backend.
All the experiments are conducted on Dell Alienware with Intel(R) Core(TM) i9-9980XE CPU, 128G memory, NVIDIA 1080Ti, and Ubuntu 16.04.

We conduct all the experiments using two public data sets, FMNIST \cite{fmnist2017}, and CIFAR10 \cite{cifar2009}. 
We report the experimental results on CIFAR10 in this section and the results on FMNIST in Appendix~\ref{appx:fmnist}.

The CIFAR10 data set contains natural images in 10 classes. It consists of a training set of 50,000 examples 
and a test set of 10,000 examples. We partition all the training examples randomly into two halves, with use one half for training the teacher models and the other half for distilling the student models.
For each data set the feature vectors are normalized to the range $[0, 1]$.

In all experiments, we use ResNet18 \cite{he2016a}. All models are trained or distilled for 100 epochs to guarantee convergence. The models with the best testing accuracy during training/distillation are retained.

\subsection{A Case Study}\label{sec:case}

We conduct a case study to demonstrate the watermarking mechanism in \wmname{}. We first train one watermarked teacher model and one non-watermarked teacher model using the first half of the training data, and then distill one student model from each teacher model using the second half of the training data.
To train the watermarked teacher model, we set the signal amplitude $\varepsilon = 0.05$ and the watermark key $K=(f_w, i^*, \mathbf{v}^0)$ with $f_w= 30.0$, $i^*=0$ and $\mathbf{v}^0$ a unit random vector.
For extraction, we set $q_{min}$ to be the first quartile of all $\mathbf{q}_{i^*}(\mathbf{x})$ values for 1,000 randomly selected training examples whose ground truth is class $i^*$.
Code for this case study is available online~\footnote{\url{https://developer.huaweicloud.com/develop/aigallery/notebook/detail?id=2d937a91-1692-4f88-94ca-82e1ae8d4d79}}.

We analyze the output of the teacher models and the student models for both the time and frequency domains in Figure~\ref{fig:case_study} for three cases.
In Figures~\ref{fig:case_study}(a), (b), and (c) for the three cases, we plot $\mathbf{q}_{i^*}(\mathbf{x})$ vs.\ $p(\mathbf{x})$ in time domain for both the teacher model and the student model in the left graph, and $P(f)$ vs.\ $f$ in the frequency domain for the student model in the right graph.

In the first case, Figure~\ref{fig:case_study}(a) shows the results for the non-watermarked teacher model and the student model.
There is no sinusoidal signal in the output for either the teacher model or the student model at frequency $f_w$ with projection vector $\mathbf{v}^0$. 

In the second case, Figure~\ref{fig:case_study}(b) shows the results for the watermarked teacher model and the student model.  The accuracy loss of the watermarked teacher model is within $1\%$ of the accuracy of the unwatermarked teacher model in Figure~\ref{fig:case_study}(a).
We extract the output of the watermarked teacher model and the student model using the watermark key $K$. 
The output of the teacher follows an almost perfect sinusoidal function and the output of the student model is close to the output of the teacher model in the time domain. 
In the frequency domain, the student model has a very prominent peak at frequency $f_w$. 
This observation validates Remark~\ref{rem:remark2} in Section~\ref{sec:theoretical} when $N=1$.

In the last case, we replace $\mathbf{v}^0$ by a different unit random vector $\mathbf{v}^1$ in the watermark key $K$ to extract the output of the watermarked teacher model and the student model. 
The results are shown in Figure~\ref{fig:case_study}(c). 
The output of both the teacher model and the student model is almost indiscernible from noise.
Thus, there is no significant peak for the output of the student model in the power spectrum at frequency $f_w$.  This observation validates Remark~\ref{rem:remark3} in Section~\ref{sec:theoretical}.

\begin{figure*}[t]%
    \centering
    \subfigure[Single Teacher]{
    \includegraphics[width = 0.23\textwidth]{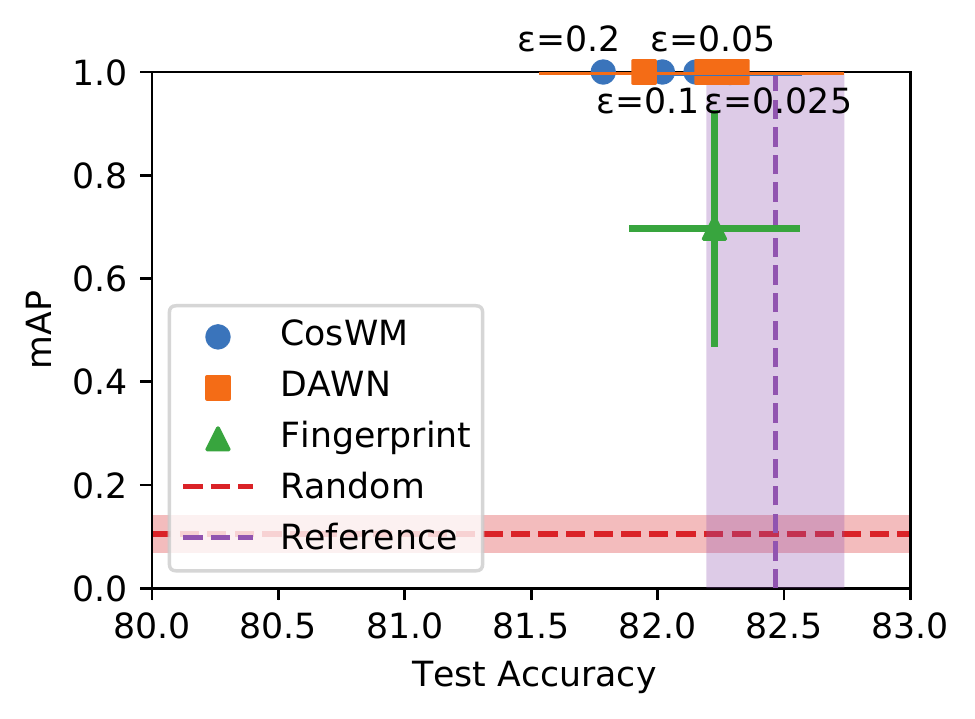}
    }
    \subfigure[2-model Ensemble]{
    \includegraphics[width = 0.23\textwidth]{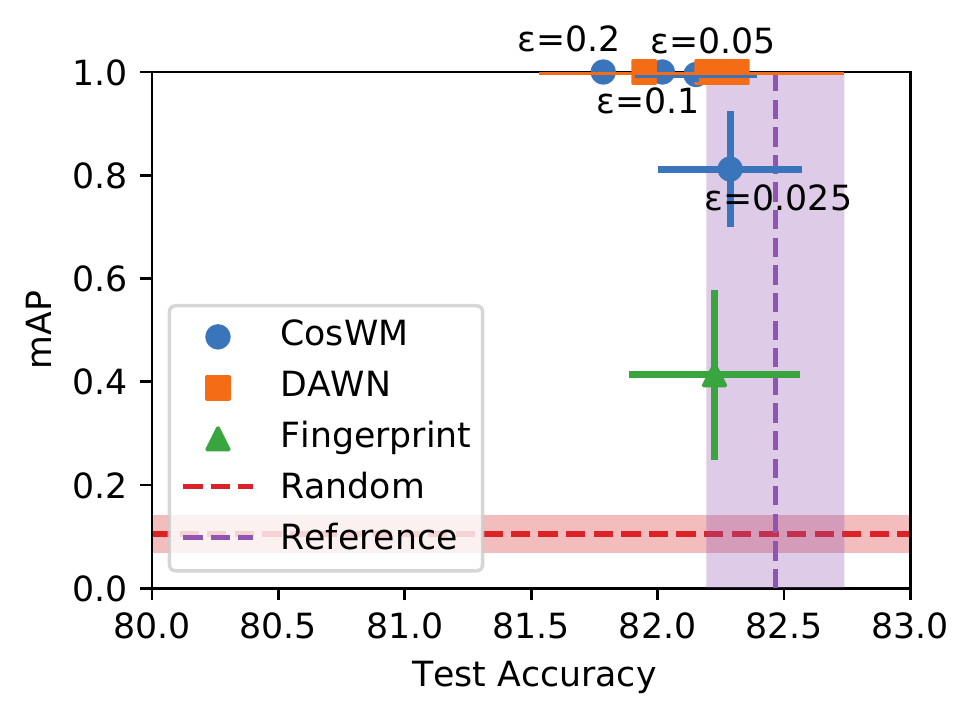}
    }
    \subfigure[4-model Ensemble]{
    \includegraphics[width = 0.23\textwidth]{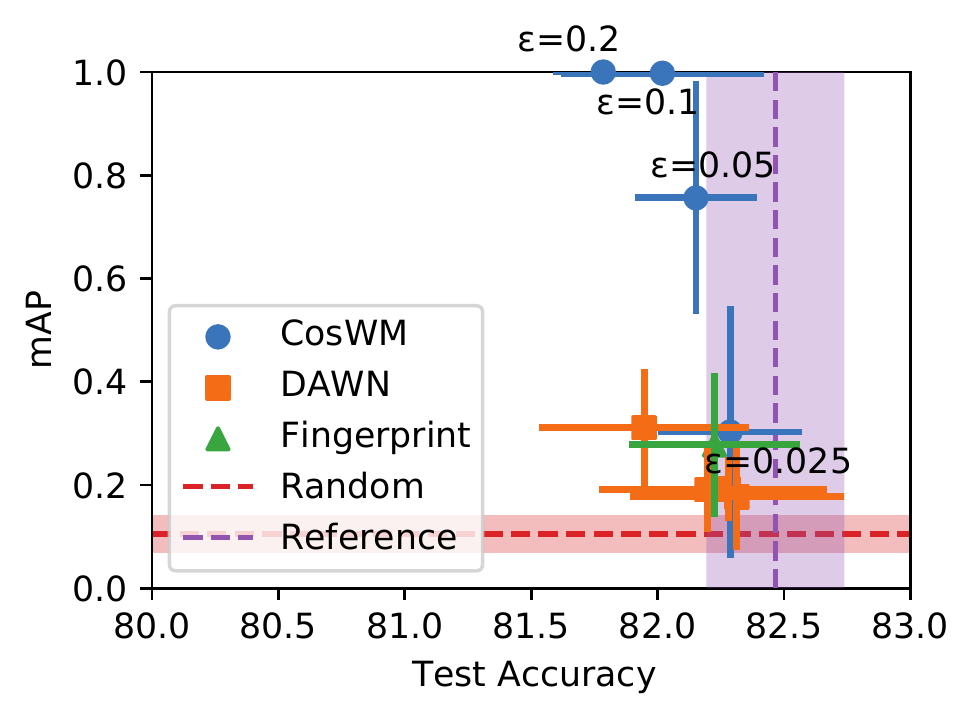}
    }
    \subfigure[8-model Ensemble]{
    \includegraphics[width = 0.23\textwidth]{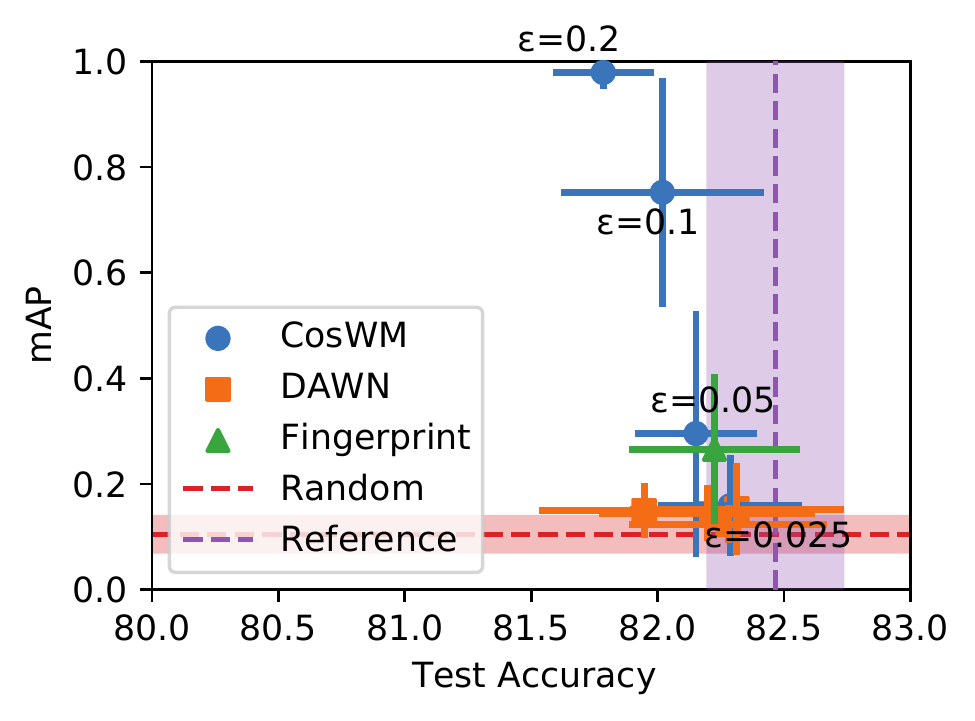}
    }
    \caption{mAP of \wmname{}, DAWN, and Fingerprinting under different parameter values as a function of accuracy of the watermarked model.
    Each watermarked model is part of an ensemble of teacher models and is the only watermarked model within that ensemble.}
    \label{fig:map_vs_acc_1wm}
\end{figure*}

\begin{figure*}[t]%
    \centering
    \subfigure[2-model Ensemble]{
    \includegraphics[width = 0.23\textwidth]{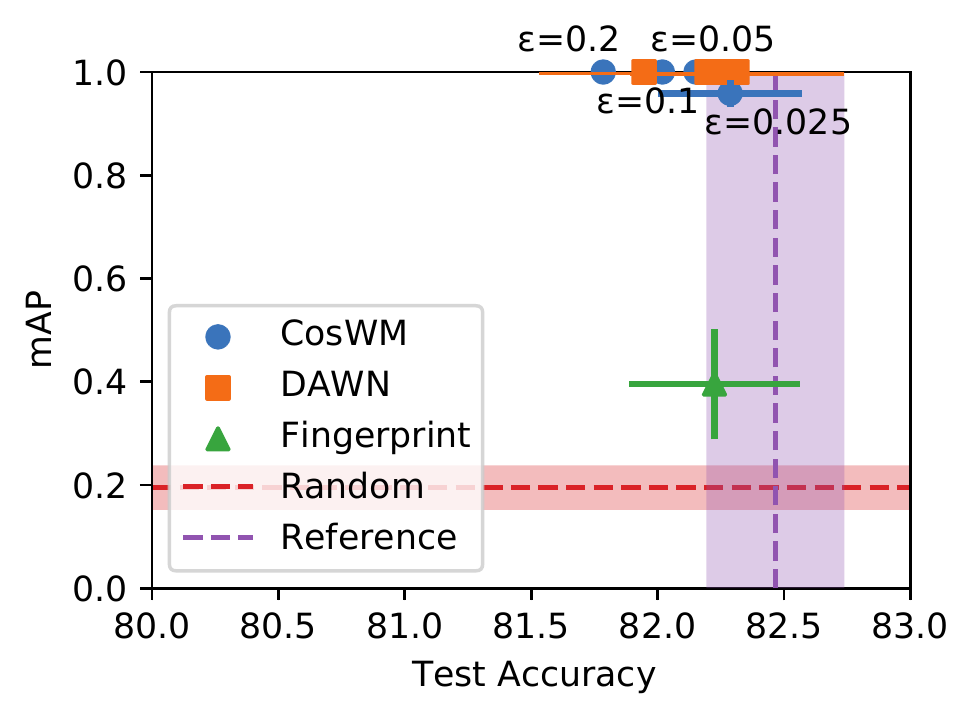}
    }
    \subfigure[4-model Ensemble]{
    \includegraphics[width = 0.23\textwidth]{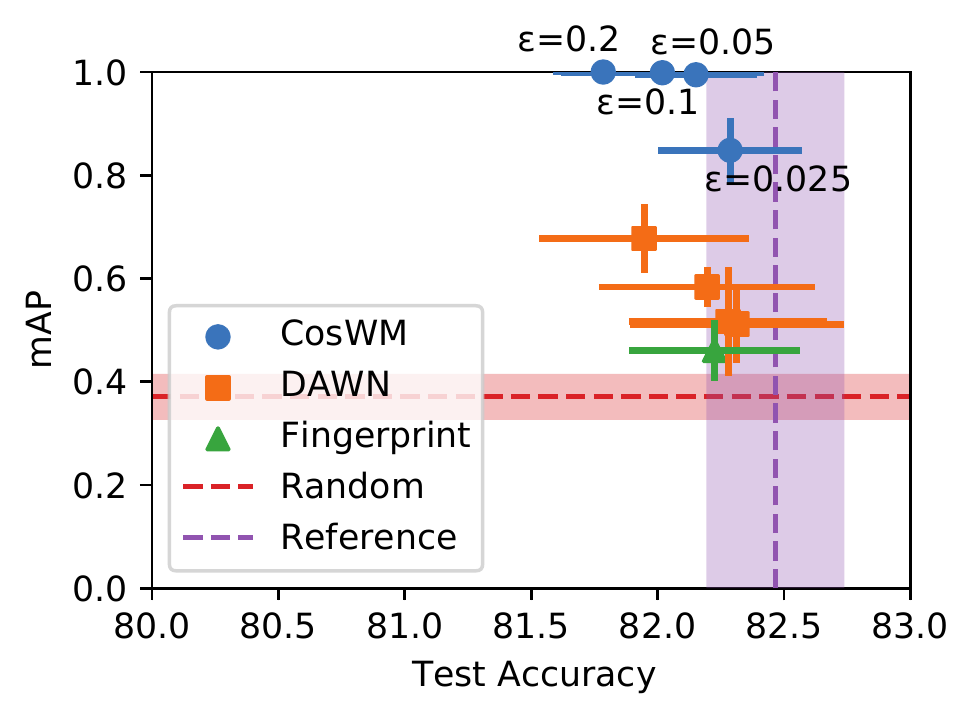}
    }
    \subfigure[6-model Ensemble]{
    \includegraphics[width = 0.23\textwidth]{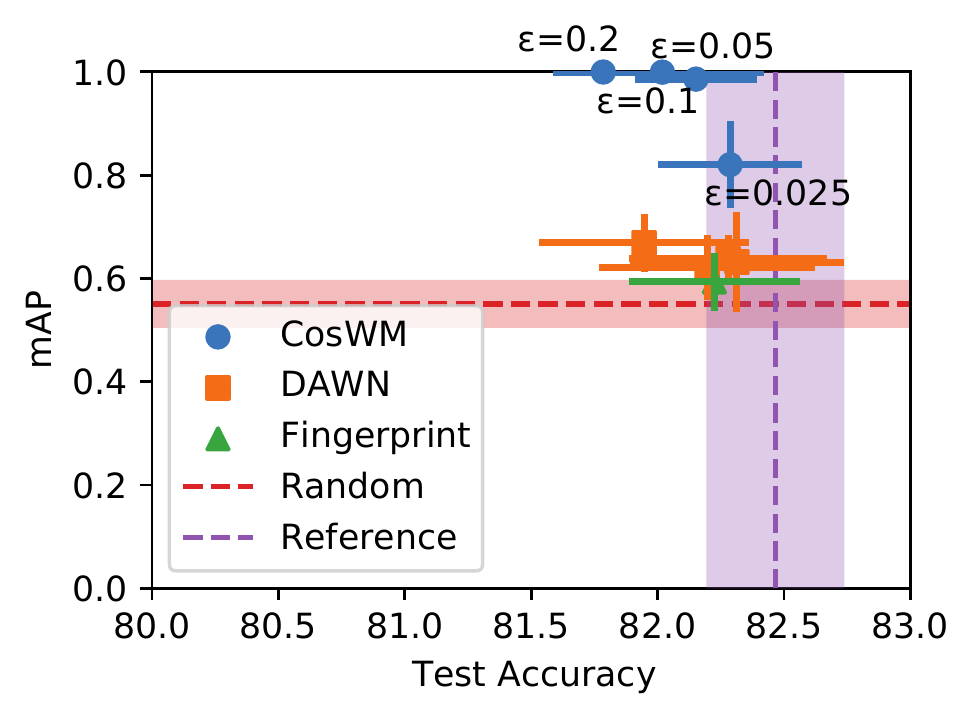}
    }
    \subfigure[8-model Ensemble]{
    \includegraphics[width = 0.23\textwidth]{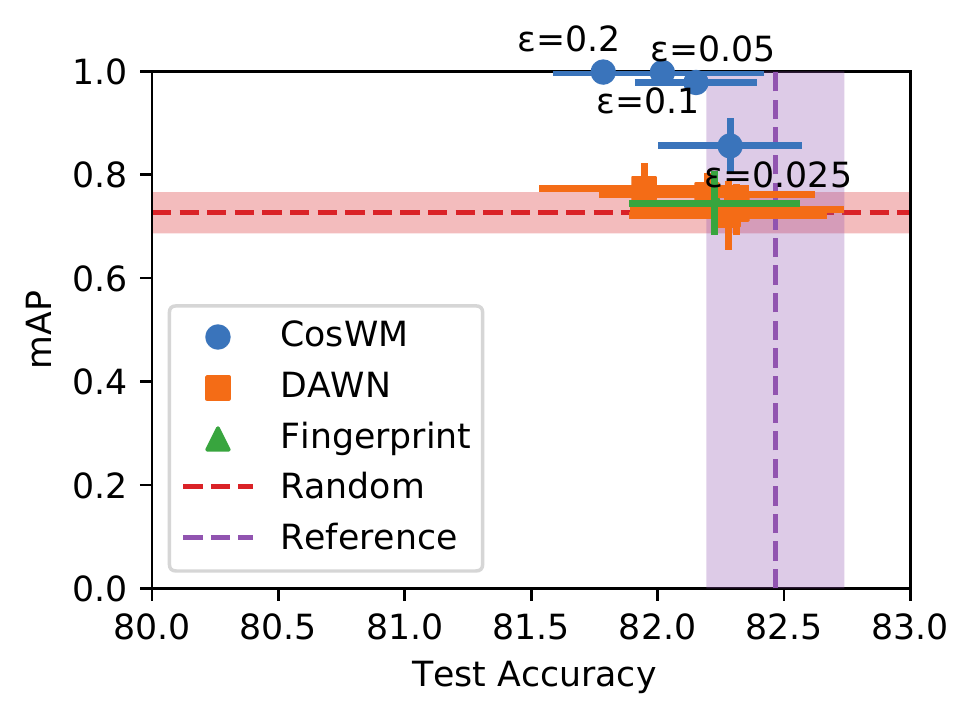}
    }
    \caption{mAP of \wmname{}, DAWN and Fingerprinting under different parameter values as a function of accuracy of the watermarked model.
    Each watermarked model is part of an ensemble of teacher models where every model is watermarked.}
    \label{fig:map_vs_acc_awm}
\end{figure*}

\subsection{Protection with a Single Watermarked Teacher}
\label{sec:single}

To compare \wmname{} with DAWN and Fingerprinting in protecting watermarked teacher models, we set up a series of ranking tasks with different ensemble size $N$.
In each ranking task, we have $10$ student models distilled from the watermarked teacher model (positive student models) and $100$ student models not distilled from the watermarked teacher model (negative students).
For different methods, we use their own watermark signal strength values to rank those $110$ students.
Specifically, we use $P_{snr}$ defined in Equation~\eqref{psnr} for \wmname{}, the fraction of matching watermark predictions for DAWN, and the fraction of matching fingerprint predictions for Fingerprinting.
To evaluate the performance of all three methods, we compute the average precision (AP) for each ranking task and repeat each task for all $10$ watermarked models to calculate the mean average precision (mAP) and its standard deviation.

For all three methods, we use the first half of the training data to train $10$ unwatermarked teacher models with different initialization and $10$ teacher models with different watermark or fingerprint keys. We tune the parameters to make sure that the accuracy losses of all watermarked teacher models are within $1\%$ of the averaged accuracy of the unwatermarked teacher models. 
To create a ranking task with $110$ student models, for every watermarked teacher model we assemble it with $N-1$ randomly selected unwatermarked teacher models to distill $10$ student models with different initialization. 
In addition, we train $10$ independent student models with ground truth labels and different initialization. 
The above process gives us $10$ positive and $100$ negative student models for each watermarked teacher model.

For \wmname{}, all watermarked teacher models have the same frequency $f_w=30.0$ and target class $i^*=0$, but have $10$ different unit random projection vectors $\mathbf{v^0}, \ldots, \mathbf{v^9}$.
We set $q_{min}$ to the median of all $\mathbf{q}_{i^*}$ values and vary the watermark amplitude $\varepsilon$ in $0.025$, $0.05$, $0.1$, and $0.2$.
For DAWN, we vary the fraction of watermarked input $\tau$ in $0.0005$, $0.001$, $0.002$, and $0.005$.
For Fingerprinting, we generate one single set of fingerprint input per teacher model using parameter $\varepsilon_{fp} = 0.095$, which results in a large enough set of fingerprints with the best conferrability score.
During extraction, all fingerprint input and labels are tested on a model to compute the fingerprint strength value for ranking.

Figure~\ref{fig:map_vs_acc_1wm} shows the results on the CIFAR10 data set for different ensemble size values $N=1,2,4,8$.
In this figure, we plot the mAP scores as a function of the average teacher model accuracy.
As a baseline, we add a \textit{Random} method that ranks all student models randomly, whose mAP and standard deviation is represented by the horizontal red dashed line.
The vertical purple dashed line shows the average and standard deviation of the accuracy of the unwatermarked teacher models.

As shown for both CosWM and DAWN in Figure~\ref{fig:map_vs_acc_1wm}, a stronger watermark will negatively affect model performance.
A model owner must consider this effect when tuning the watermark.

When the ensemble size is small, i.e., $1$ and $2$, the best mAP of \wmname{} and DAWN are generally comparable, and are both significantly larger than that of Fingerprinting, as shown in Figures~\ref{fig:map_vs_acc_1wm}(a) and (b).
When the ensemble size is larger, i.e., $4$ and $8$, the best mAP of \wmname{} is significantly larger than that of DAWN and Fingerprinting, whose watermarked model is consistently outnumbered, as shown in Figures~\ref{fig:map_vs_acc_1wm}(c) and~\ref{fig:map_vs_acc_1wm}(d).
This superior performance of \wmname{} is due to our watermark signal design that is robust to ensemble distillation.
When the ensemble size increases, \wmname{} needs a larger $\varepsilon$ to keep mAP high.
This confirms the discussion in Remark~\ref{rem:remark2} in Section~\ref{sec:theoretical}.   

In addition, we observe a trade-off between ensemble size and mAP when choosing different signal amplitude  $\varepsilon$ for \wmname{} and different fraction $\tau$ of watermarked input for DAWN. We analyze the effect of the amplitude parameter $\varepsilon$ in more details in Appendix~\ref{appx:ens_size_graphs}. 

\nop{Moreover, we show the results on FMNIST with this experimental setting in Appendix~\ref{appx:extra}.}

\subsection{Protection with Multiple Watermarked Teachers}
\label{sec:multiple}

We compare \wmname{} with DAWN and Fingerprinting in assembling only watermarked teacher models to train a student model by undertaking another series of ranking tasks for different ensemble sizes $N$.
We train 10 watermarked teacher models as described in Section~\ref{sec:single}, and assemble $10$ sets of teacher models for each ensemble size in a round-robin manner.
The training of all other models and watermark settings in this experiment remain exactly the same as described in Subsection~\ref{sec:single}.
As a result, in an $N$-ensemble teacher model experiment, each ranking task associated to a teacher model has $10N$ positive and $110 - 10N$ negative student models.

Figure~\ref{fig:map_vs_acc_awm} shows the results on the CIFAR10 data set for different ensemble size values, i.e., $N=2, 4, 6, 8$. It is plotted in the same way as in Figure~\ref{fig:map_vs_acc_1wm}, described in Section~\ref{sec:single}. Similar to the previous experiments, we also add the \textit{Random} baseline to provide a lower bound performance. 

The accuracy losses of all watermarked models are within $1\%$ of the average accuracy of all unwatermarked teacher models.
When the ensemble size is small, i.e., $N=2$, the best mAP of \wmname{} and DAWN are generally comparable to each other, and are both significantly larger than that of Fingerprinting, as shown in Figure~\ref{fig:map_vs_acc_awm}(a).
However, \wmname{} has a significantly higher best mAP for larger ensemble sizes, i.e., $N=4, 6, 8$, as shown in Figures~\ref{fig:map_vs_acc_awm}(b), (c) and (d).
This shows that \wmname{} watermarks are generally unaffected by other watermarks in a teacher ensemble and confirms the possibility of detecting watermarks if the ensemble features multiple watermarked teacher models as discussed in Section~\ref{extension_multi}.

We also observe a similar trade-off between ensemble size and mAP when choosing different signal amplitude  $\varepsilon$ for \wmname{} and different fraction $\tau$ of watermarked input for DAWN.
This is further analyzed in Appendix~\ref{appx:ens_size_graphs}.

\nop{Moreover, we show the results on FMNIST with this experimental setting in Appendix~\ref{appx:extra}.}


\section{Conclusion}
\label{sec:conclusion}
In this paper, we tackle a novel problem of protecting neural network models against ensemble distillation.
We propose \wmname{}, an effective method relying on a signal embedded into all output of a watermarked model, and therefore transferring the signal to training data for student models.
We prove that the embedded signal in \wmname{} is strong in a well-trained student model by providing lower and upper bounds on the watermark strength metric.
In addition, \wmname{} can be extended to identify student models distilled from an ensemble featuring multiple watermarked models.
Our extensive experiments demonstrate the superior performance of \wmname{} in providing models defense from ensemble distillation.

\bibliography{egbib_aaai_clean}

\clearpage
\appendix
\begin{center}
{\Large \bf Appendix}
\end{center}
\setcounter{section}{0}
In this appendix, we provide the proofs for Theorem 1 in Subsection~\ref{proof} and a lemma on properties of the modified softmax outputs in Subsection~\ref{proof-lemma}. 
In addition, we show more extensive experimental results in Subsection \ref{supexp}.

\section{Proofs}
\subsection{Proof of Theorem \ref{bound}}\label{proof}
\begin{proof}
We first prove the left inequality of \eqref{eq:bound}. By using the triangular inequality and the fact that $\mathbf{s}^*(\bar{D})$
is the best sinusoidal fit for $\bar{D}$, we have
\begin{align}
    \chi^2_f(D) & \geq \sum_{l=1}^L \left[ \mathbf{\bar q}^l_{i^*} -[\mathbf{s}^*(D)]_l \right]^2 - \sum_{l=1}^L \left[ \mathbf{\bar q}^l_{i^*} -\mathbf{q}^l_{i^*} \right]^2 \\
    &  \geq \sum_{l=1}^L \left[ \mathbf{\bar q}^l_{i^*} -[\mathbf{s}^*(\bar{D})]_l \right]^2 - L_{SE} \\
   &  = {\chi}^2_f(\bar{D})  - L_{SE}.
\end{align}
Combining the above inequality with equation (\ref{eq:LSdef}) we obtain the left inequality of \eqref{eq:bound}.

Next we prove the right inequality of \eqref{eq:bound}. Since $\mathbf{s}^*(\hat{D})$ and $\mathbf{s}^*(\widetilde{D})$ are points from the graph of sinusoidal functions with same frequency $f$, then $\frac{1}{N}\mathbf{s}^*(\hat{D}) +  \frac{N-1}{N}\mathbf{s}^*(\widetilde{D})$ is also a set of points from the graph of a sinusoidal function.
As $\mathbf{s}^*(D)$ is the best sinusoidal fit for $D$, we have
\begin{align}
   \chi^2_f & (D) = \sum_{l=1}^L \left[ \mathbf{q}^l_{i^*} - [{\mathbf{s}}^*(D)]_l \right]^2 \\
& \leq \sum_{l=1}^L \left[ \mathbf{q}^l_{i^*} - \left[\frac{1}{N}\mathbf{s}^*(\hat{D}) +  \frac{N-1}{N}\mathbf{s}^*(\widetilde{D})\right]_l \right]^2. 
\end{align}
By using the above inequality and triangular inequality, we get the following:

\begin{align}
    &\chi^2_f (D)  \leq \sum_{l=1}^L \left[ \mathbf{q}^l_{i^*} {-} \frac{1}{N}[ \mathbf{s}^*(\hat{D})]_l {-}  \frac{N-1}{N}[\mathbf{s}^*(\widetilde{D})]_l \right]^2\\
    & \leq \sum_{l=1}^L \left[ \mathbf{q}^l_{i^*} - \frac{1}{N}\hat{\mathbf{q}}^l_{i^*} - \frac{N-1}{N}\mathbf{\widetilde q}^l_{i^*} \right]^2 +\\
    & \quad {\sum_{l=1}^L} \left[ \frac{1}{N}\hat{\mathbf{q}}^l_{i^*} {+} \frac{N-1}{N}\mathbf{\widetilde q}^l_{i^*} {-} \frac{1}{N}[ \mathbf{s}^*(\hat{D})]_l {-} \frac{N-1}{N}[ \mathbf{s}^*(\widetilde{D})]_l \right]^2\nonumber\\
& \leq L_{SE} + \\
& \sum_{l=1}^L \left[ \frac{1}{N}\left[\hat{\mathbf{q}}^l_{i^*} - [ \mathbf{s}^*(\hat{D})]_l\right]+ \frac{N-1}{N}\left[\mathbf{\widetilde q}^l_{i^*}  - [ \mathbf{s}^*(\widetilde{D})]_l\right] \right]^2\nonumber\\
    & \leq L_{SE} +  \frac{1}{N^2} {\chi}^2_f(\hat{D}) + \left(\frac{N-1}{N}\right)^2{\chi}^2_f((\widetilde{D}).
\end{align}
Combining the above inequality with equation (\ref{eq:LSdef}) we obtain the right inequality of \eqref{eq:bound}.

\end{proof}

\subsection{Lemma 1: Modified Softmax Properties}
\label{proof-lemma}
\begin{lem}
\label{lemma-1}
Let $\mathbf{q}$ be the softmax output of a model $R$, then the modified softmax $\hat{\mathbf{q}}$, as defined in \eqref{eq:modified_softmax} satisfies Property \ref{prop:softmax}.
\end{lem}

\begin{proof}
    By the definition of softmax (\ref{eq:softmax}), for all $i$ we have 
    \begin{gather}
        0 \leq \mathbf{q}_i \leq 1,\\
        -1 \leq \mathbf{a}_i(\mathbf{x}) \leq 1.
    \end{gather}
    Therefore, when $i = i^*$, we have
    \begin{equation}
        0 \leq \mathbf{q}_i + \varepsilon (1 + \mathbf{a}_i(\mathbf{x})) \leq 1 + 2 \varepsilon,
    \end{equation}
    and then 
   \begin{equation}
      0 \leq \frac{\mathbf{q}_i + \varepsilon (1 + \mathbf{a}_i(\mathbf{x}))}{1 + 2 \varepsilon} \leq 1. 
   \end{equation}
    When $i \neq i^*$, since $m\geq 2$, we have
    \begin{equation}
         0 \leq \mathbf{q}_i + \frac{\varepsilon (1 + \mathbf{a}_i(\mathbf{x}))}{m - 1} \leq 1 + \frac{2 \varepsilon}{m-1} \leq 1 + 2 \varepsilon,
    \end{equation}
    and then
    \begin{equation}
       0 \leq \frac{\mathbf{q}_i + \frac{\varepsilon (1 + \mathbf{a}_i(\mathbf{x}))}{m - 1}}{1 + 2 \varepsilon} \leq 1. 
    \end{equation}
    Thus, $\hat{\mathbf{q}}$ satisfies clause 1 of Property \ref{prop:softmax}.
    
    To prove clause 2 of Property \ref{prop:softmax}, we use the fact that $\mathbf{a}_{i^*} + \mathbf{a}_{j \neq i^*} = 0$ and obtain
    \begin{align}
        & \sum_i^m \hat{\mathbf{q}}_i = \frac{\mathbf{q}_{i^*} + \varepsilon(1 + \mathbf{a}_{i^*})}{1 + 2\varepsilon} + \sum_{j \neq i^*} \frac{\mathbf{q}_{j} + \frac{\varepsilon(1 + \mathbf{a}_j)}{m-1}}{1 + 2\varepsilon}\\
        & = \left(\sum_{j=1}^m \frac{\mathbf{q}_i}{1 + 2 \varepsilon}\right) + \left(\sum_{j\neq i^*} \frac{\varepsilon (1 + \mathbf{a}_{i^*} + 1 + \mathbf{a}_{j})}{(m-1)(1 + 2 \varepsilon)}\right)\\
        & = \frac{1}{1 + 2 \varepsilon}  + \frac{2 \varepsilon}{1 + 2 \varepsilon}\\
        & = 1.
    \end{align}
    Thus, clause  2 of Property \ref{prop:softmax} is also satisfied.
\end{proof}

\section{More Experiment Results}
\label{supexp}

\subsection{FMNIST Results}
\label{appx:fmnist}

We include here the results of the experiments described in Sections~\ref{sec:single} and~\ref{sec:multiple} for models trained on the FMNIST data set.
The FMNIST data set contains fashion item images in 10 classes. It consists of a training set of 60,000 examples 
and a test set of 10,000 examples.
We partition all the training examples randomly into a teacher half and a student half just as we did with CIFAR10.

Figure \ref{fig:map_vs_acc_1wm_fmnist} shows the result of the single watermark experiment described in section~\ref{sec:single} using the FMNIST data set.
The figure plots the mAP of the different watermarks against the test accuracy of the watermarked model.
All of the experiment parameters and baselines are the same as shown in the equivalent CIFAR10 results shown in  Figure~\ref{fig:map_vs_acc_1wm}.

The accuracy losses of all the watermarked models are on average negligible when compared to equivalent unwatermarked models.
Similar to the CIFAR10 results, the mAP of \wmname{} and DAWN watermarks are comparable for smaller ensemble sizes, as shown in Figures~\ref{fig:map_vs_acc_1wm}(a) and (b), and the \wmname{} watermarks are significantly stronger for larger ensembles, as shown in Figures~\ref{fig:map_vs_acc_1wm}(c) and (d).
This solidifies our claim that the design of \wmname{} makes it the most robust method against ensemble distillation.

\begin{figure*}[t]%
    \centering
    \subfigure[Single Teacher]{
    \includegraphics[width = 0.226\textwidth]{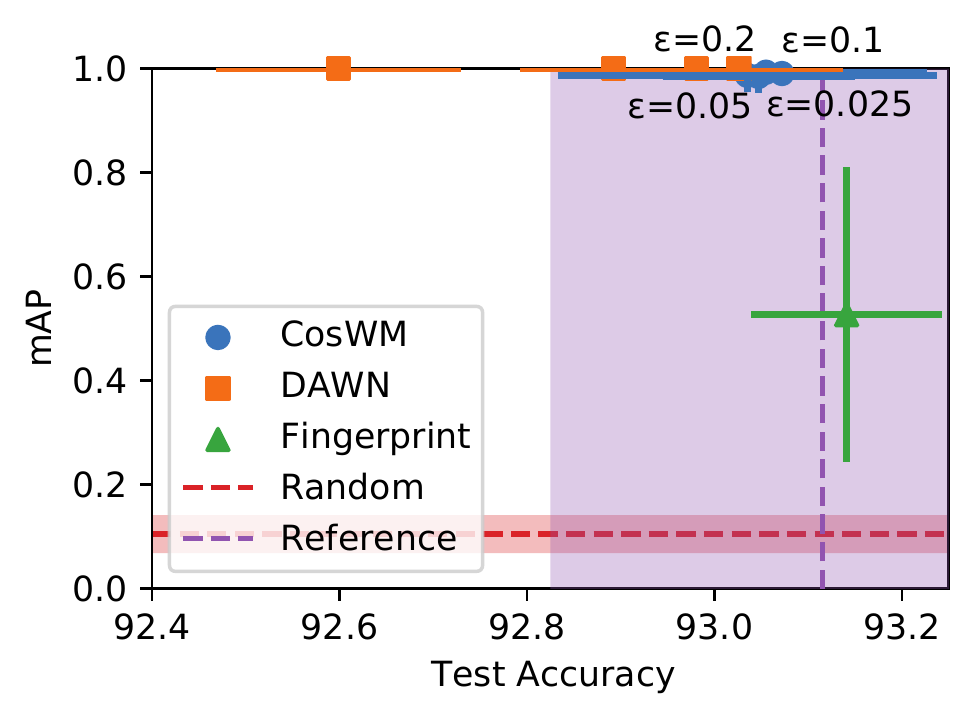}
    }
    \subfigure[2-model Ensemble]{
    \includegraphics[width = 0.226\textwidth]{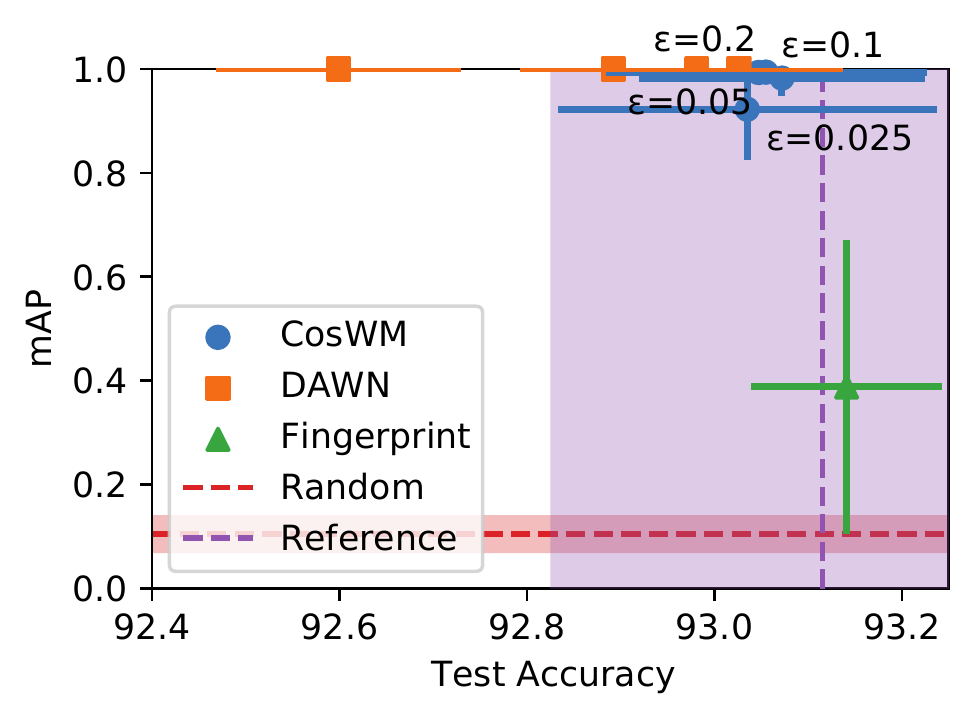}
    }
    \subfigure[4-model Ensemble]{
    \includegraphics[width = 0.226\textwidth]{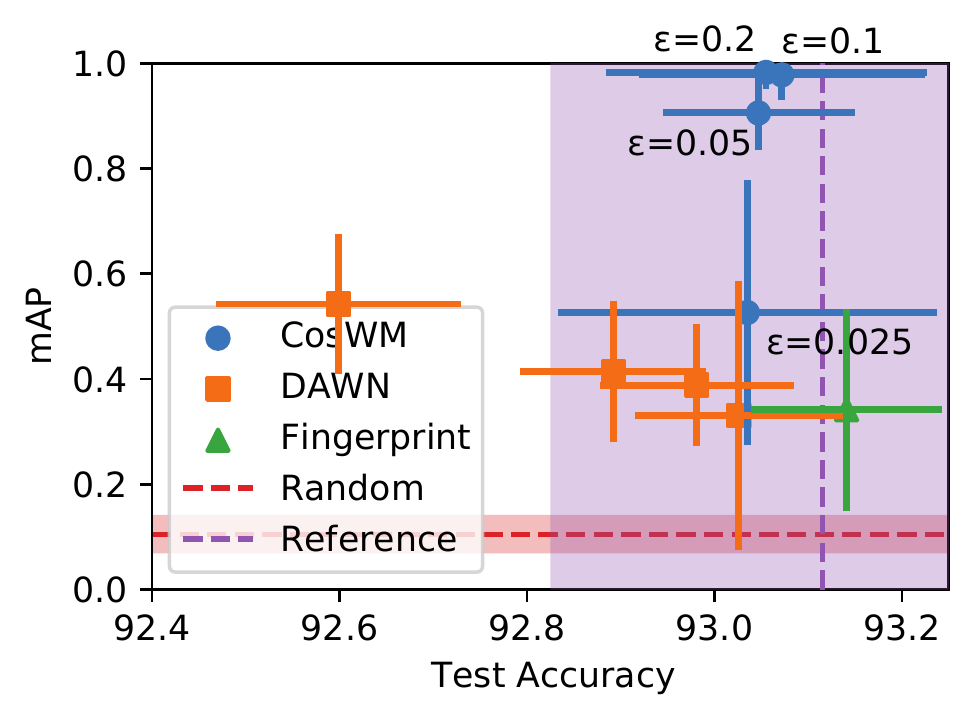}
    }
    \subfigure[8-model Ensemble]{
    \includegraphics[width = 0.226\textwidth]{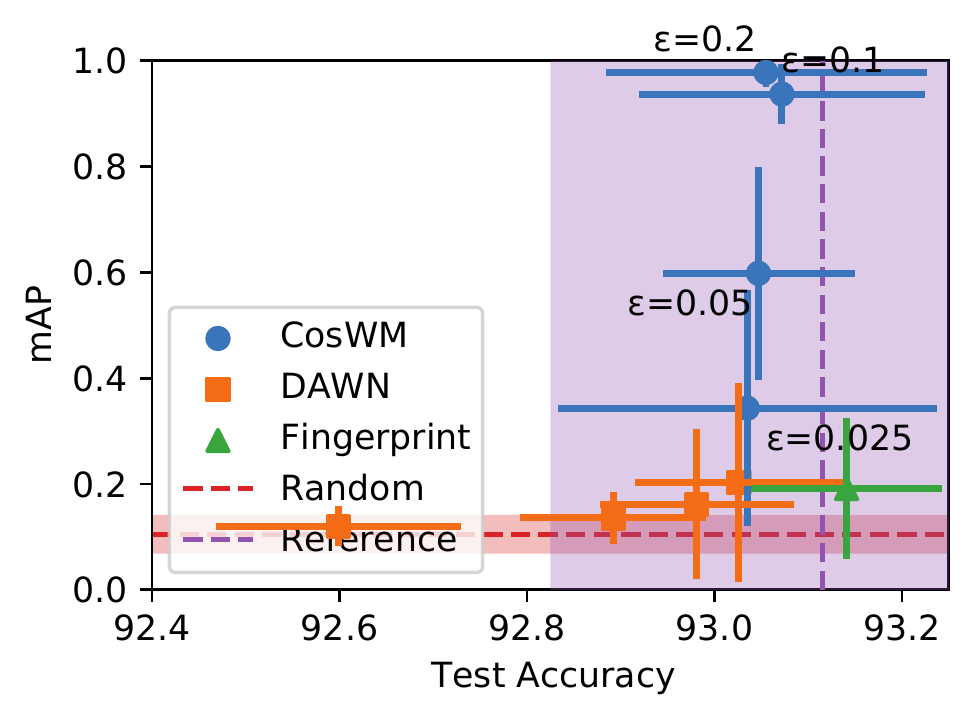}
    }
    \caption{mAP of \wmname{}, DAWN, and Fingerprinting under different parameter values as a function of accuracy of the watermarked model trained on the FMNIST data set.
    Each watermarked model is part of an ensemble of teacher models and is the only watermarked model within that ensemble.}
    \label{fig:map_vs_acc_1wm_fmnist}
\end{figure*}

\begin{figure*}[t]%
    \centering
    \subfigure[2-model Ensemble]{
    \includegraphics[width = 0.226\textwidth]{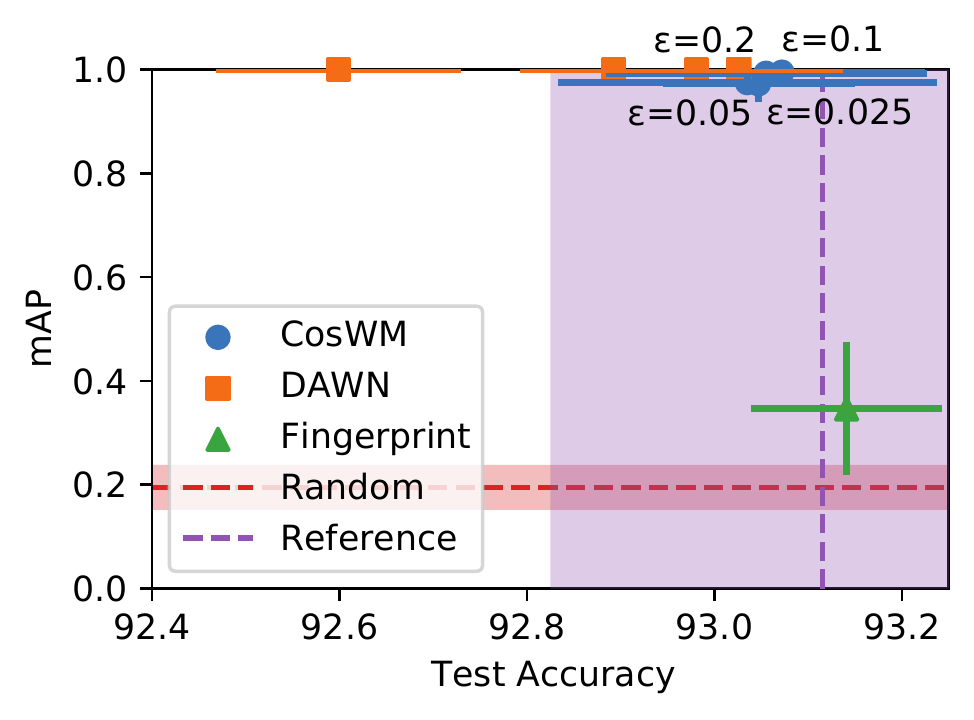}
    }
    \subfigure[4-model Ensemble]{
    \includegraphics[width = 0.226\textwidth]{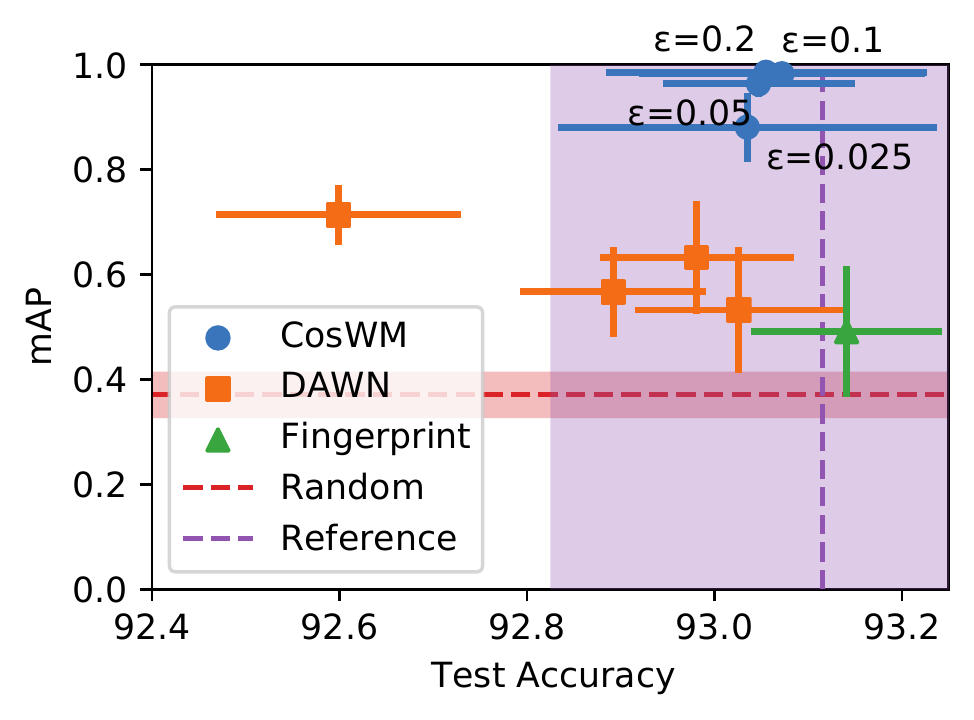}
    }
    \subfigure[6-model Ensemble]{
    \includegraphics[width = 0.226\textwidth]{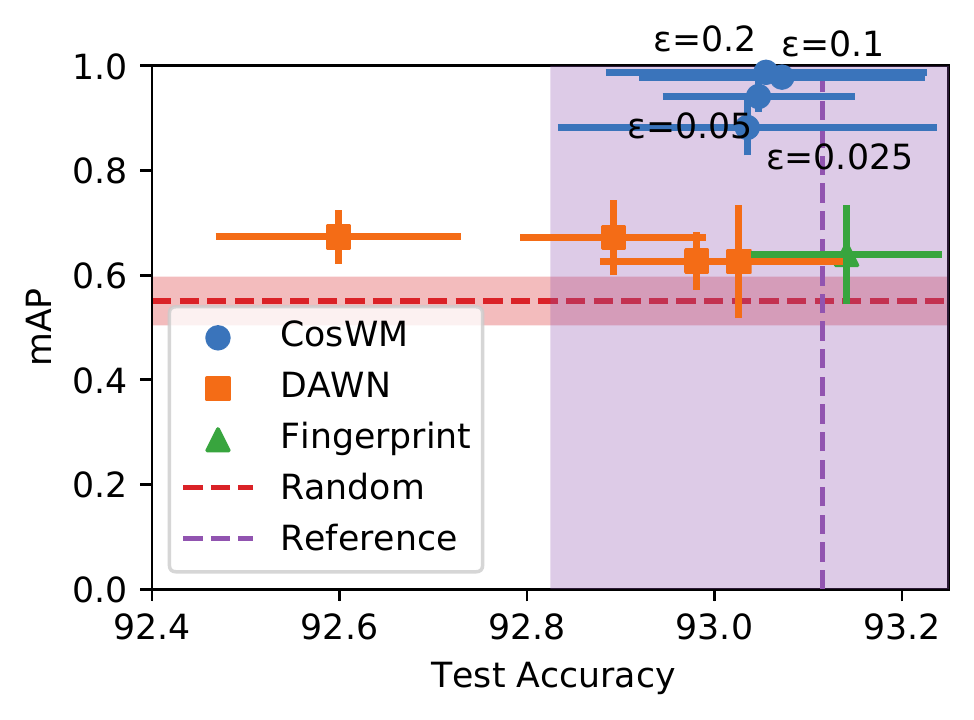}
    }
    \subfigure[8-model Ensemble]{
    \includegraphics[width = 0.226\textwidth]{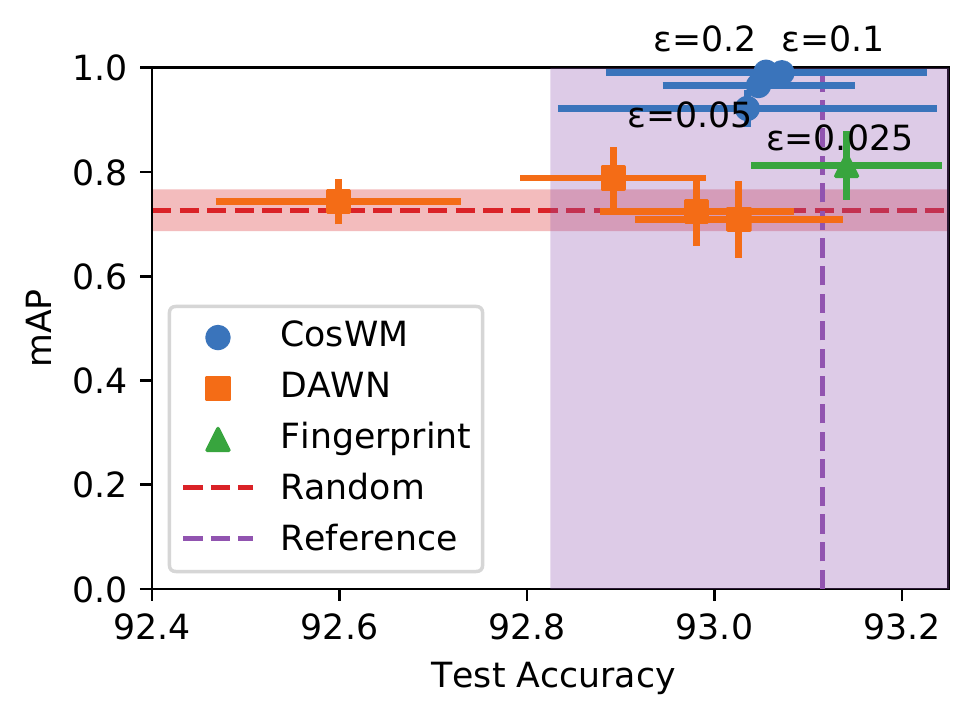}
    }
    \caption{mAP of \wmname{}, DAWN and Fingerprinting under different parameter values as a function of accuracy of the watermarked model trained on the FMNIST data set.
    Each watermarked model is part of an ensemble of teacher models where every model is watermarked.}
    \label{fig:map_vs_acc_awm_fmnist}
\end{figure*}

Figure \ref{fig:map_vs_acc_awm_fmnist} shows the result of the multiple watermark experiment described in section~\ref{sec:multiple} on the FMNIST data set.
Once more, the figure plots mAP against watermarked model test accuracy for many watermark methods.
All of the experiment parameters and baselines are the same as shown in the equivalent CIFAR10 results shown in Figure~\ref{fig:map_vs_acc_awm}.

The accuracy losses of all the watermarked models are still negligible compared to unwatermarked models.
Similar to the CIFAR10 results and the single model experiments The watermark produced by \wmname{} is significantly more robust than DAWN or Fingerprinting when used with larger ensemble sizes of $N=4, 6, 8$.
This solidifies our claims that \wmname{} has the ability to discern each distinct watermark featured in an ensemble from a student model.

\subsection{Amplitude Parameter Analysis}
\label{appx:ens_size_graphs}

To demonstrate the flexibility and tuning capabilities of \wmname{}, we analyze the effect of the watermark parameters to the identification performance.
We start with the amplitude parameter $\varepsilon$.
In general, because a larger amplitude signal results in a larger $P_{snr}$ value, a watermark signal with a larger amplitude should make it more likely to be successfully extracted.
Therefore we expect watermarks with higher $\varepsilon$ to have a higher mAP.

We can observe from Figures \ref{fig:map_vs_acc_1wm}-\ref{fig:map_vs_acc_awm_fmnist} that this is generally the case for both \wmname{} and DAWN.
A stronger signal usually infers higher mAP scores, at the cost of lower watermarked model performance.
Therefore, while setting the amplitude of the watermark signal, one must consider striking a balance between model performance and watermark performance.
A similar balancing act exists when setting the ratio of trigger inputs $\tau$ in a DAWN watermark.

\begin{figure}
    \centering
    \begin{minipage}[b]{0.3\textwidth}
        \includegraphics[width = 1.00\textwidth]{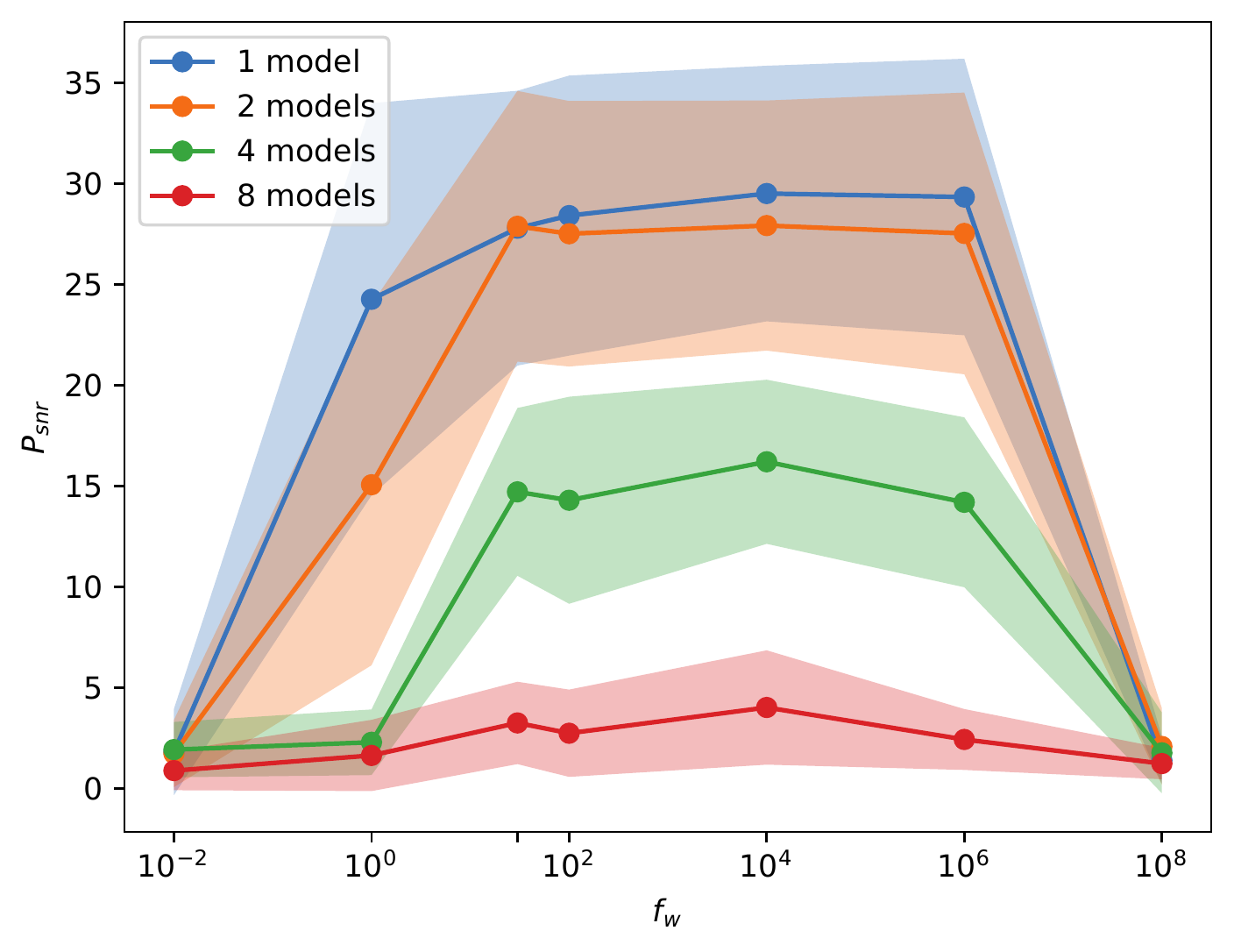}
        \caption{Watermark strength $P_{snr}$ as a function of signal frequency $f_w$ for different ensemble sizes.
        Each ensemble contains only one watermarked model.}
        \label{fig:k_study_v2}
    \end{minipage}
\end{figure}

\begin{figure*}[t]%
    \centering
    \subfigure[Single Teacher]{
    \includegraphics[width = 0.226\textwidth]{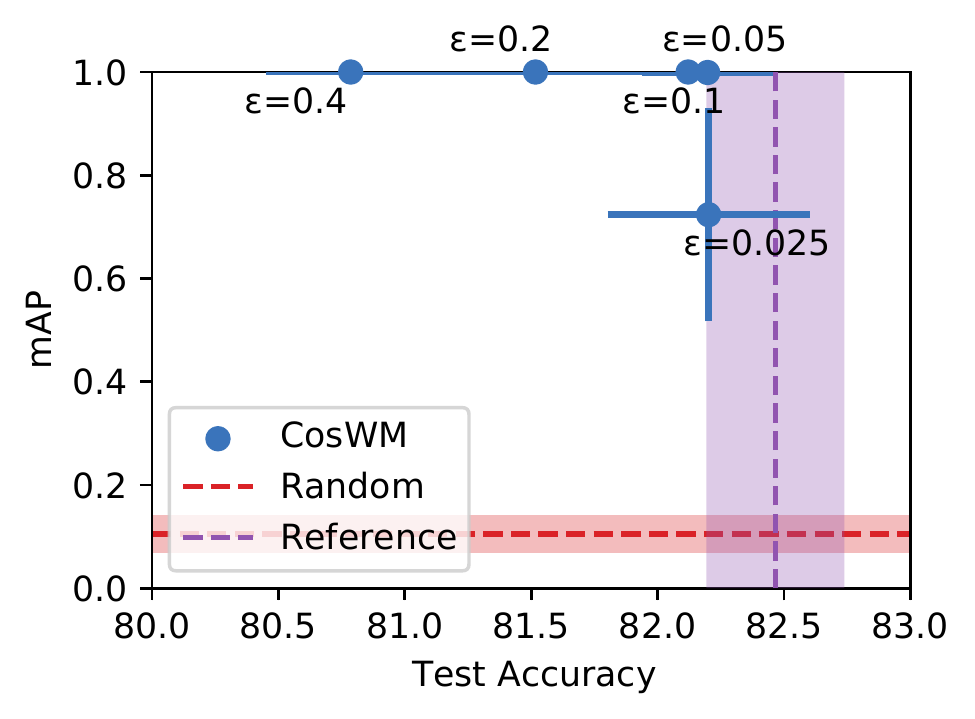}
    }
    \subfigure[2-model Ensemble]{
    \includegraphics[width = 0.226\textwidth]{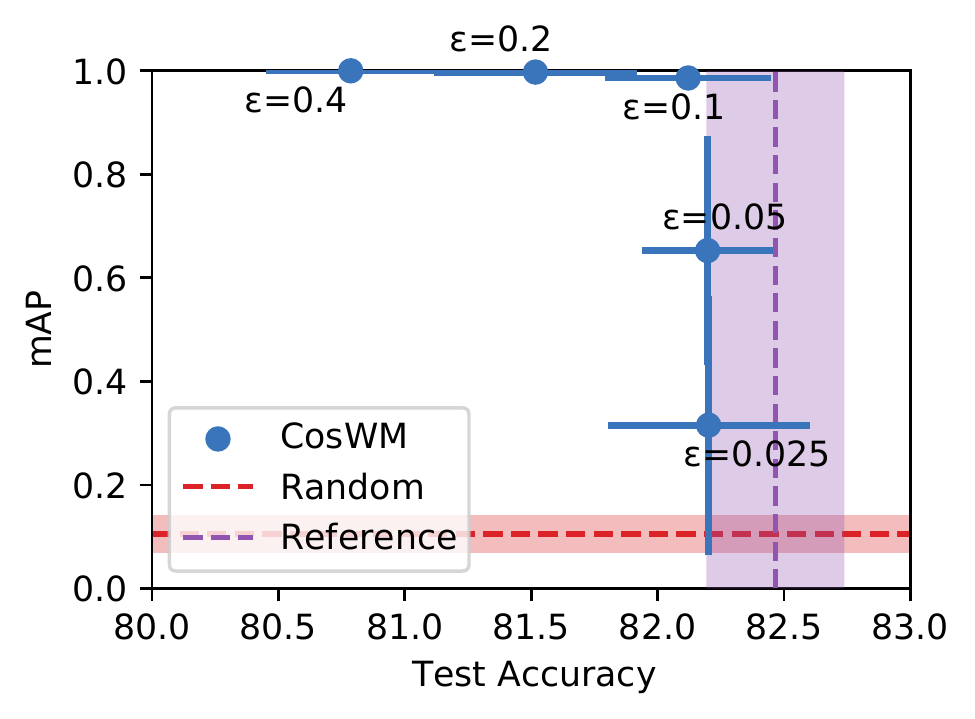}
    }
    \subfigure[4-model Ensemble]{
    \includegraphics[width = 0.226\textwidth]{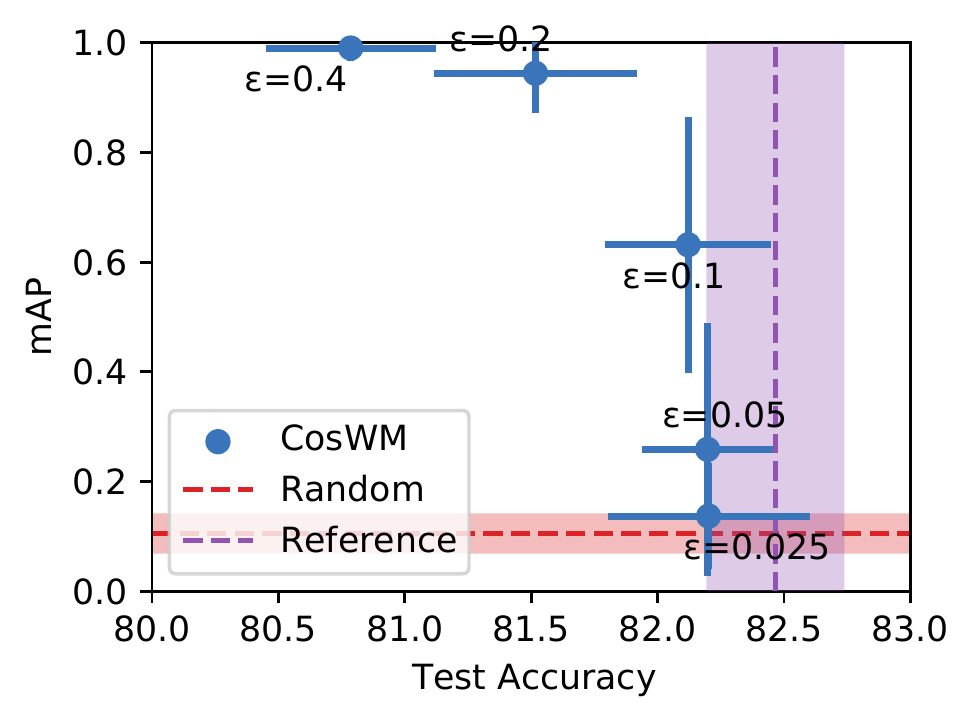}
    }
    \subfigure[8-model Ensemble]{
    \includegraphics[width = 0.226\textwidth]{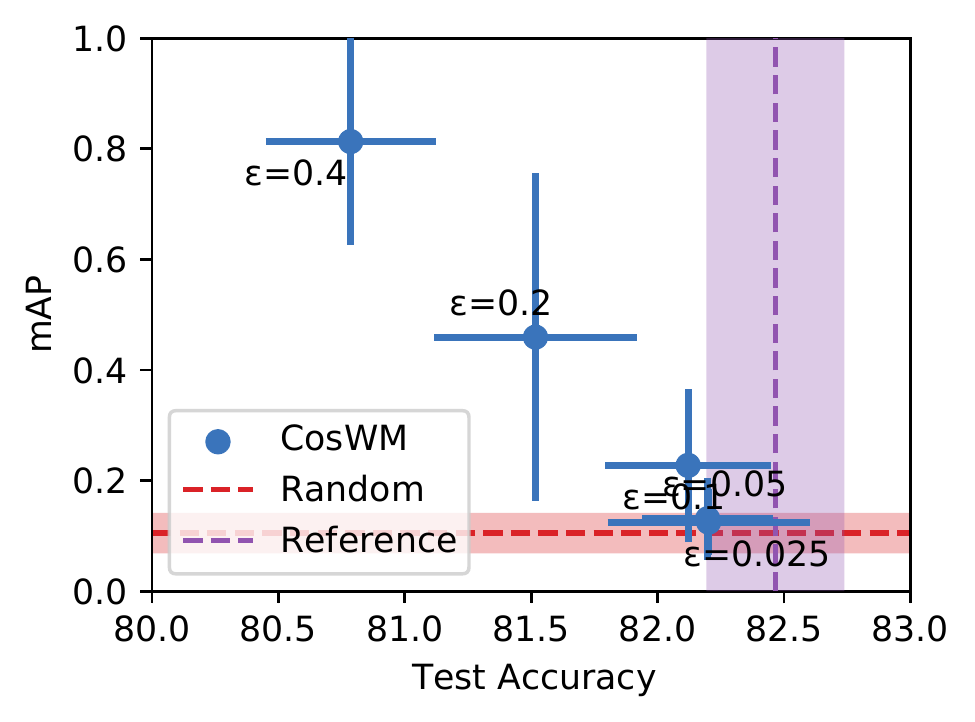}
    }
    \caption{mAP of \wmname{} under different parameter values as a function of accuracy of the watermarked model.
    All models are trained on the CIFAR10 data set.
    Each watermarked model is part of an ensemble of teacher models and is the only watermarked model within that ensemble.
    Students are distilled with a combination of KL loss and cross entropy loss.}
    \label{fig:map_vs_acc_1wm_mixed}
\end{figure*}

\subsection{Signal Frequency Parameter Analysis}
\label{appx:frequency}
As a demonstration of \wmname{}'s versatility, we show that the method can be effective at extracting the signal function for a wide range of frequencies $f_w$.

For the different values of $f_w$ of $0.01$, $1.0$, $30.0$, $100.0$, $10^4$, $10^6$, $10^8$, we train 5 watermarked teacher models and distill 10 student models per an ensemble which includes only one watermarked teacher.
We use parameter values $\varepsilon=0.2$, and five randomly generated unit projection vectors $\mathbf{v}^0, ..., \mathbf{v}^4$.
For each student model, we compute the $P_{snr}$ value of the matching teacher watermark at the frequency $f_w$ and compare its value as we change $f_w$.
We also compute the standard deviation of $P_{snr}$ to use as a confidence interval.
We repeat the experiment for ensemble sizes $N=2, 4, 8$, where each model ensemble assembles only one watermarked teacher model with randomly generated unwatermarked teacher models as described in Subsection \ref{sec:single}.

Figure \ref{fig:k_study_v2} shows the results of this analysis.
We clearly see that \wmname{} performs very well for a very wide range of $f_w$ values, but becomes less effective if $f_w$ is too small or too large.
Frequency values of $f_w=30, 100, 10^4, 10^6$ all result in $P_{snr}$ values generally above $5.0$ for $N=1, 2, 4$.
From our previous experiments under similar settings in Subsection \ref{sec:single}, such high values will always result in an mAP of exactly or extremely close to $1.0$.
Even with $N=8$, we observe a noticeable increase in $P_{snr}$ between $f_w=30$ and $10^4$, which will increase mAP to relatively high scores, compared to another method like DAWN or Fingerprinting, as shown in Figures \ref{fig:map_vs_acc_1wm}(d) and \ref{fig:map_vs_acc_1wm_fmnist}(d).
This shows that on top of having a variety of choice for the projection vector $\mathbf{v}$, \wmname{} also offers a large range in the choice of frequency $f_w$.

\subsection{Selecting Extraction Parameters}
\label{appx:extraction_parameters}
Here we explain how to properly select the extraction parameters $q_{min}$ and $\delta$ in the extraction algorithm described in Subsection~\ref{extraction}.

The value of $q_{min}$ should be selected so that we only keep points inside the range of the signal function.
However, this can be difficult to determine systematically, as the amplitude of the signal in the student will decrease when the ensemble size increases.
To remedy this, we choose $q_{min}$ to be proportional to the median or the first quartile of the output values, which allows us to determine its value automatically.
From our experience, the first quartile value works well in the single teacher case, while the median value better generalizes for larger ensembles.

In applications where few data points have high confidence outputs, such as data sets with many classes, one may instead filter outputs below a set $q_{max}$ threshold. 
The idea is to substitute high confidence outputs that would be close to $1$, as in Figure~\ref{fig:keyidea}(a) for an unwatermarked model with low confidence outputs that would be close to $0$.
Applying the perturbation in either case should result in a cosine signal with a clearer peak in its power spectrum.

The value of $\delta$ should be chosen as to only contain a significant part of the power spectrum peak.
In numerical computations, frequency values are determined by an uniformly spaced sample, and choosing $\delta$ is equivalent to choosing a fixed number of closest frequencies to $f_w$.
Peaks can vary in width depending on the experiment, but choosing a narrower window will generally have little impact on the computed value of $P_{snr}$.
We have found that a sufficient window can contain only a few (less than 10) frequency samples to generalize well across all experiments.

\subsection{Combination of KL loss and cross entropy loss}
\label{appx:truelabel}
In this subsection, we conduct experiments on \wmname{} to extract watermarks from student models distilled with a combination of KL loss and cross entropy loss, in a similar way to several common distillation processes such as \cite{ba2014a, bucilua2006a4, hinton2015a}. 
We use the same training settings as the CIFAR10 experiments described in Subsection \ref{sec:single}.
In each task, we distill 100 student models with an equally weighted combination of KL loss from the ensemble's outputs and cross entropy loss from the ground truth labels.
We set $q_{min}$ to the median of the $\mathbf{q}_{i^*}$ values of the sampled inputs with ground truth label $i^*$, and vary the watermark amplitude $\varepsilon$ in $0.025$, $0.05$, $0.1$, $0.2$, and $0.4$.

Figure~\ref{fig:map_vs_acc_1wm_mixed} shows the results on the CIFAR10 data set for different ensemble size values, i.e., $N=1, 2, 4, 8$.
Results are plotted in the same way as in Figure~\ref{fig:map_vs_acc_1wm}, shown in Subsection~\ref{sec:single}.
Similar to the previous experiments, we also add a \textit{Random} baseline to provide a lower bound performance for all the methods. 

The accuracy losses of all watermarked models are within $1\%$ of the average accuracy of all unwatermarked teacher models except when $\varepsilon=0.4$.
\wmname{} has very high mAP for ensemble sizes $N=1, 2, 4$, as shown in Figure~\ref{fig:map_vs_acc_1wm_mixed}(a), (b), and (c). However, for the ensemble size $N=8$, \wmname{} needs to sacrifice more accuracy to achieve high mAP as shown in Figure~\ref{fig:map_vs_acc_1wm_mixed}(d). The major reason for this observation is that the watermark signal has been diluted by the cross-entropy loss since the student models are distilled with an equally weighted combination of KL loss and cross entropy loss.

\end{document}